\documentclass[11pt,twoside]{article}

\usepackage[usenames,x11names]{xcolor}
\usepackage[english]{babel}
\usepackage{latexsym,amssymb,amsmath,amsthm}
\usepackage{graphicx}
\usepackage[dvips,pagebackref]{hyperref}
\usepackage[a4paper]{geometry}
\usepackage{breakurl}
\usepackage{colortbl}
\usepackage{multirow}
\usepackage{rotating}

\hypersetup{%
	pdftitle={Counting Triangulations and other Crossing-free Structures via Onion Layers},%
	pdfauthor={Victor Alvarez, Karl Bringmann, Radu Curticapean, Saurabh Ray},%
	pdfkeywords={Algorithmic geometry, Algorithms, Geometry, Crossing-free structures, Counting algorithms, Triangulations, Matchings, Spanning cycles, Experiments, Hardness results},%
	pdfdisplaydoctitle=true,%
	pdfstartview=Fit,%
	linktocpage,
}

\newtheorem{theorem}{Theorem}
\newtheorem{lemma}{Lemma}

\newtheorem{definition}{Definition}

\newcommand{\R}{\mathbb{R}}

\newcommand{\T}{\mathcal{T}}
\newcommand{\setp}{P}
\newcommand{\F}{\mathcal{F}}
\newcommand{\Conv}{\textup{CH}}
\renewcommand{\SS}{\mathcal{S}}
\newcommand{\CDTA}{\overline{\triangle}}

\title{Counting Triangulations and other Crossing-free Structures via Onion Layers}
\author{Victor Alvarez\thanks{Fachrichtung Informatik, Universit\"{a}t des Saarlandes, {\tt{alvarez@cs.uni-saarland.de}}. Partially Supported by CONACYT-DAAD of M\'{e}xico.} \and Karl Bringmann\thanks{Max-Planck-Institut f\"ur Informatik, {\tt{kbringma@mpi-inf.mpg.de}}.} \and Radu Curticapean\thanks{Fachrichtung Informatik, Universit\"{a}t des Saarlandes, {\tt{curticapean@cs.uni-saarland.de}}} \and Saurabh Ray\thanks{Max-Planck-Institut f\"ur Informatik. {\tt{saurabh@mpi-inf.mpg.de}}.}}
\date{\today}

\begin{document}
\maketitle

\begin{abstract}
Let $\setp$ be a set of $n$ points in the plane. A crossing-free structure on $\setp$ is a plane graph with vertex set $P$. Examples of crossing-free structures include triangulations of $P$, spanning cycles of $P$, also known as polygonalizations of $P$, among others. There has been a large amount of research trying to bound the number of such structures. In particular, bounding the number of (straight-edge) triangulations spanned by $P$ has received considerable attention. It is currently known that \emph{every} set of $n$ points has at most $O(30^{n})$ and at least $\Omega(2.43^{n})$ triangulations. However, much less is known about the algorithmic problem of counting crossing-free structures of a given set $P$. For example, no algorithm for counting triangulations is known that, on all instances, performs faster than enumerating all triangulations. 

In this paper we develop a general technique for computing the number of crossing-free structures of an input set $P$. We apply the technique to obtain algorithms for computing the number of triangulations, matchings, and spanning cycles of $P$. The running time of our algorithms is upper bounded by $n^{O(k)}$, where $k$ is the number of \emph{onion layers} of $P$. In particular, for $k = O(1)$ our algorithms run in polynomial time. In addition, we show that our algorithm for counting triangulations is \emph{never} slower than $O^{*}(3.1414^{n})$, even when $k = \Theta(n)$. Given that there are several well-studied configurations of points with at least $\Omega(3.464^{n})$ triangulations, and some even with $\Omega(8^{n})$ triangulations, our algorithm asymptotically outperforms  any  enumeration algorithm for such instances, and it has better worst-case behavior than the recent algorithm shown in~\cite{sweep-line}, which also beats enumeration in those instances. In fact, it is widely believed that any set of $n$ points must have at least $\Omega(3.464^{n})$ triangulations. If this is true, then our algorithm is strictly sub-linear in the number of triangulations counted. We also show that our techniques are general enough to solve the {\sc Restricted-Triangulation-Counting-Problem}, which we prove to be $W[2]$-hard in the parameter $k$. This implies a ``no free lunch'' result: In order to be fixed-parameter tractable, our general algorithm must rely on additional properties that are specific to the considered class of structures.
\end{abstract}

\section{Introduction}

Let $\setp\subset\R^{2}$ be a set of $n$ points. A crossing-free structure on $\setp$ is a plane graph whose vertex set is precisely $\setp$. Examples of such crossing-free structures, whose proper definition will be given later on, are triangulations, spanning cycles, matchings, spanning trees, etc. Thus one can naturally ask: (\oldstylenums{1}) What are upper and lower bounds on the number of such structures over all possible sets of $n$ points on the plane? or (\oldstylenums{2}) Given $\setp$, how can the number of such geometric objects be computed? The search for bounds, the first question, has spawned a large amount of research over almost 30 years, starting with an upper bound of $10^{13n}$ on the number of crossing-free graphs on every set of $n$ points, see~\cite{Ajtai19829}. This bound implies that the size of \emph{each} class of crossing-free structures can be upper-bounded by $c^{n}$, with $c \in \mathbb{R}$ depending on the particular class. Since then, research has focused on tightening the upper and lower bounds on $c$. For example, in the case of spanning cycles, it is currently known that $c \leq 54.55$, see~\cite{DBLP:journals/corr/abs-1109-5596}, and a configuration where $c\geq 4.65$ is known, see~\cite{DBLP:journals/comgeo/GarciaNT00}. Thus, every set of $n$ points has at most $O(54.55^n)$ spanning cycles. For triangulations, \cite{DBLP:journals/combinatorics/SharirS11} provides the bound $c \leq 30$, and \cite{DBLP:journals/jct/SharirSW11} provides $c \geq 2.4$. The interested reader can visit~\cite{asheffer,erikdemaine} for an up-to-date list of bounds on other classes of crossing-free structures. The references therein give a good account of all listed bounds. 
	
	The second question on crossing-free structures, which was mentioned above, is of algorithmic flavor since we consider the problem of \emph{computing} the number of crossing-free structures of a particular class for a \emph{given} input set $\setp$. This problem is closely related to that of sampling crossing-free structures of the class uniformly at random. That is, if $\setp$ spans, say $t$ spanning cycles, we want to sample every spanning cycle with probability $1/t$. A first approach to the counting problem would be to produce \emph{all} elements of the class, using methods for enumeration (see e.g. \cite{DBLP:journals/dcg/KatohT09}), and then simply count the number of elements. This has the obvious disadvantage that the total time spent will be, at best, linear in the number of elements counted. By the first part, this number is in general exponential in the input size.  Thus an important question is whether we can count crossing-free structures of a given class in time sub-linear in the number of elements counted. Only for the super class of \emph{all} plane graphs of $\setp$ this is known to be \emph{always} possible, see~\cite{DBLP:conf/birthday/RazenW11}, while for the class of all triangulations it is known to be \emph{sometimes} possible using a recent algorithm shown in~\cite{sweep-line}. There are nevertheless other algorithms to count triangulations that are reported to be faster than enumeration~\cite{ray-seidel,DBLP:conf/compgeom/Aichholzer99}, but that have no theoretical runtime guarantees.

\section{Our contribution}\label{c-tri:sections:contribution}

Although the algorithm to count triangulations presented in~\cite{sweep-line} could potentially \emph{always} do the counting in sub-linear time in the number of elements counted, and thus beat brute force enumeration, its running time might still be pretty large. This is because its running time depends linearly on the number of T-paths that the algorithm encounters during its execution, and there are configurations having at least $\Omega(4^{n})$ T-paths, see~\cite{DBLP:journals/comgeo/DumitrescuGPW01}. We suggest the paper~\cite{DBLP:conf/compgeom/Aichholzer99} by O. Aichholzer, where T-paths were first introduced, to find out more about them.

In this paper we present yet another new algorithm to count triangulations. Along with this new algorithm we also present algorithms to count the elements of the classes of spanning cycles and matchings of $\setp$ respectively. It is important to keep in mind that, so far, no algorithm is known that \emph{always} beats enumeration on those classes. 

In order to state the results we present in this paper we need the following definitions:

\begin{definition}
	Let $S$ be a crossing-free structure on $\setp$. Thus:
	\begin{itemize}
		\item $S$ is called a triangulation of $\setp$ if and only if the boundary of the unbounded face of $S$ coincides with the convex hull, $\Conv(\setp)$, of $\setp$, and \emph{every} bounded face of $S$ is an empty triangle.
		\item $S$ is called a matching of $\setp$ if and only if \emph{every} vertex of $S$ has degree at most one.
		\item $S$ is called a spanning cycle of $\setp$ if and only if $S$ is single simple polygon with $n$ sides whose vertex set is $\setp$.
	\end{itemize}
\end{definition}

We will denote by $\F_{T}(\setp), \F_{M}(\setp)$ and $\F_{C}(\setp)$ the class of \emph{all} triangulations, matchings, and spanning cycles of $\setp$ respectively.

\begin{definition}[Onion layers]\label{c-tri:def:onion-layers}
	Let $\setp$ be a set of $n$ points on the plane and let $\Conv(\setp)$ denote its convex hull. We define the \emph{onion layers} of $\setp$ as follows: The first onion layer $\setp^{(1)}$ of $\setp$ is $\Conv(\setp)$. For $i > 1$, the $i$-th onion layer $\setp^{(i)}$ of $\setp$ is defined inductively as $\Conv\left(\setp\setminus\bigcup_{j = 1}^{i-1} \setp^{(j)}\right)$. By ``number of onion layers of $\setp$'' we mean the number of \emph{non-empty} onion layers of $\setp$.
\end{definition}

Observe that the number of onion layers of \emph{any} non-degenerate set of $n$ points is \emph{always} at most $\left\lceil\frac{n}{3}\right\rceil$. We are now able to state our results.

\subsection{The result on counting triangulations}

\begin{theorem}\label{c-tri:theorems:triple-paths}
	Let $\setp$ be as before and let $k$ be its number of onion layers. Then the exact value of $|\F_{T}(\setp)|$ can be computed in time $O\left(k^{2}\cdot g\left(\frac{n}{k}\right)^{n}\right)$, where $g(x) = \left(\frac{x^{3} + 3x^{2} + 2x + 2}{2}\right)^{\frac{1}{x}}$. Since $k~\leq~\left\lceil\frac{n}{3}\right\rceil$, this bound never exceeds $O^{*}(3.1414^{n})$. This running time can alternatively be bounded by $n^{O(k)}$, which is polynomial for constant $k$.
\end{theorem}

The algorithm of the previous theorem has better worst-case behavior than the one presented in~\cite{sweep-line}, which is $O^{*}\left(9^{n}\right)$. Moreover, it has other nice properties:

\begin{itemize}
	\item It is the \emph{first} algorithm to be known that can compute the exact value of $|\F_{T}(\setp)|$ in \emph{polynomial time} in at least some non-trivial cases.
	\item As stated before, for \emph{every} set of $n$ points, the size of $\F_{T}(\setp)$ can be lower-bounded by $\Omega(2.4^{n})$, but it is widely believed that this bound can be improved to $\Omega\left(\sqrt{12}^{n}\right)\approx\Omega(3.464^{n})$. If this stronger bound is true, then the algorithm of Theorem~\ref{c-tri:theorems:triple-paths} would always count triangulations in time $O^{*}(3.1414^{n}) = o(|\F_{T}(\setp)|)$. Thus setting in the positive the answer of whether or not one can \emph{always} count triangulations of set of points faster than enumerating them.
\end{itemize}

\subsection{The results on counting other crossing-free structures}

Moving away from triangulations, the other result that will be proven is the following:

\begin{theorem}\label{c-tri:theorems:triangular-paths}
	Let $\setp$ be as before and let $k$ be its number of onion layers. Then the exact values of $|\F_{M}(\setp)|$ and $|\F_{C}(\setp)|$ can be computed in $n^{O(k)}$ time.
\end{theorem}

Thus again, as long as $k = O(1)$, the algorithms of Theorem~\ref{c-tri:theorems:triangular-paths} compute the said numbers in polynomial time, which then gives a partial answer to Problem 16 of The Open Problems Project, which asks whether $|\F_{C}(\setp)|$ can \emph{always} be computed in polynomial time, see~\cite{topp}. This time, however, we are not able to prove a running time of the sort $c^{n}$ for large $k$, like in Theorem~\ref{c-tri:theorems:triple-paths}.

The general layout of the algorithms of Theorems~\ref{c-tri:theorems:triple-paths} and~\ref{c-tri:theorems:triangular-paths} is similar to the one found in~\cite{DBLP:journals/comgeo/AnagnostouC93}, where these ideas have been used for optimization problems.

\subsection{A hardness result}

Observe that the running times of the algorithms of Theorems~\ref{c-tri:theorems:triple-paths} and~\ref{c-tri:theorems:triangular-paths} can be stated as $n^{f(k)}$, for some function $f$ that does not depend on $n$. With regard to parameterized complexity it is natural to ask if these running times can be improved  to something of the sort $g(k)\cdot n^{O(1)}$, for some function $g$ independent of $n$, thus proving that our problems belong to the FPT complexity class. Which is the class of fixed-parameter tractable problems. However, the techniques involved in the algorithms of Theorems~\ref{c-tri:theorems:triple-paths} and~\ref{c-tri:theorems:triangular-paths} are general enough to solve harder problems, such as the following:

{\sc Restricted-Triangulation-Counting-Problem}\label{RTCP}: Given a set of points $\setp$ and a subset of edges $E$ over $\setp$, count the triangulations of $\setp$ that use only edges from $E$.

We also present the following hardness result:

\begin{theorem}\label{c-tri:theorems:fpt}
	The {\sc Restricted-Triangulation-Counting-Problem} is W[2]-hard if the parameter is considered to be the number of onion layers of $\setp$. This result even holds for the problem of just deciding the \emph{existence} of a restricted triangulation.
\end{theorem}
	
The book by J. Flum and M. Grohe, see~\cite{flum-grohe}, is a standard reference for Parameterized Complexity Theory, where the classes FPT and W[2] are defined. For now, however, it suffices to say that the separation FPT $\neq$ W[2] is widely believed. Thus an algorithm with a running time of the sort $g(k)\cdot n^{O(1)}$ is most likely not attainable for the {\sc Restricted-Triangulation-Counting-Problem}. This might be an indication that we may have to exploit the particular structure of the problems in order to obtain fixed-parameter tractable algorithms for counting crossing-free structures, in the general non-restricted case.

We will divide the rest of the paper as follows: In~\S~\ref{c-tri:sections:framework} we give a rough idea on how our algorithms work. We prove Theorems~\ref{c-tri:theorems:triple-paths}, \ref{c-tri:theorems:triangular-paths}, and~\ref{c-tri:theorems:fpt} in~\S~\ref{c-tri:sections:triple-paths}, \S~\ref{c-tri:sections:triangular-paths}, and~\S~\ref{c-tri:sections:hardness} respectively. In~\S~\ref{c-tri:sections:experiments} we show experiments comparing our algorithm for counting triangulations with the algorithm presented in~\cite{ray-seidel}, which is supposed to be \emph{very} fast in practice. We close the paper in~\S~\ref{c-tri:sections:conclusionsTripleP} with some conclusions.

\section{A general framework for counting crossing-free structures}\label{c-tri:sections:framework}

The overall idea of \emph{all} our algorithms can be roughly described as follows. Suppose we want to count the elements of some particular class $\F$ of crossing-free structures on $\setp$. A set $S$ of non-crossing edges on $\setp$ is called a \emph{separator} if the union of the edges in $S$ splits the interior of $\Conv(\setp)$, possibly along with $\Conv(\setp)$, into at least two regions. In such a case we will say that $S$ splits $\Conv(\setp)$ into those regions. Now assume that there exists a set $\SS$ of separators with the following properties: (\oldstylenums{1}) Every element of $\F$ contains a \emph{unique} separator $S\in\SS$, and (\oldstylenums{2}) we can ``quickly'' enumerate the members of $\SS$. With a set of separators $\SS$, the elements of $\F$ can be counted as follows: For each $S\in\SS$, let $R^S_1$, $R^S_2$, \ldots, $R^S_t$ be the regions $S$ splits $\Conv(\setp)$ into. Recursively compute the number $n^{S}_{i}$ of elements of $\F$ of each region $R^S_i$. The number of elements of $\F$ containing $S$ is then $N^S=\prod_{i=1}^{t} n^{S}_{i}$. Thus the total number of elements of $\F$ is simply $\sum_{S \in \SS} N^S$. Of course, in the recursion, a set of separators is required in each $R^S_i$, and the efficiency of the algorithm depends heavily on the choice of $\SS$. For example, one well-known family of separators $\SS$ for triangulations is the set of T-paths, see~\cite{DBLP:conf/compgeom/Aichholzer99,DBLP:journals/comgeo/DumitrescuGPW01,sweep-line}. We will introduce some other families of separators, some of them with additional properties, however, for the time being we believe that this vague description of how the algorithms work conveys the main idea appropriately.

\section{Counting triangulations using the onion layers}\label{c-tri:sections:triple-paths}

In this section we will present yet another new algorithm for counting triangulations that uses the onion layers of the given set of points $\setp$. 

Now, for any $p\in\setp$, let $\ell(p)$ denote the index of the onion layer to which $p$ belongs. Let us label the points $p\in\setp$ with distinct labels in $\{1,\ldots,n\}$ such that if $\ell(p) < \ell(q)$ then $p$ also receives a label smaller than $q$. This is clearly possible. Figure~\ref{c-tri:sections:t-paths:figs:20} shows the onion layers of a set of $17$ points and the labels assigned to them. From now on we will refer to the points of $\setp$ by their labels \emph{i.e.}, we will think of $\setp$ as the set $\{1,\ldots,n\}$ and when we say ``$p\in\setp$'', we will mean the point with label $p$. 

\begin{figure}[!hbt]
    \begin{center}
	\includegraphics[height=4.2cm]{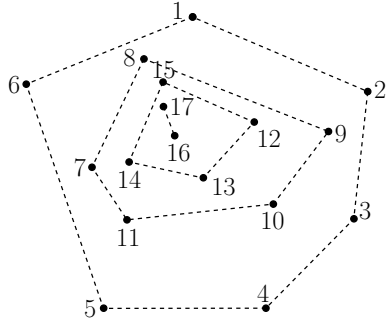}
	\caption{Four onion layers.}
	\label{c-tri:sections:t-paths:figs:20}
	\end{center}
\end{figure}

Let $T$ be any triangulation of $\setp$. For $p \in\setp\setminus\setp^{(1)}$, let $\text{sn}_T(p)$ be the smallest neighbor of $p$ in $T$. Observe that any such point $p$ has at least one neighbor $q$ such that $\ell(q) < \ell(p)$ and therefore $\text{sn}_T(p) < p$. If $p \in\setp^{(1)}$, we set $\text{sn}_T(p)=p$. When $T$ is clear from context, we will just write $\text{sn}(p)$ instead of $\text{sn}_T(p)$. We denote by $\text{sn\textcolor{black}{-path}}_T(p)$ the unique path $p=a_0,a_1,\ldots,a_m$ in $T$ such that for each $0\leq i<m$, we have that $a_{i+1}=\text{sn}(a_i)$ and $\text{sn}(a_m)=a_m$. We will also direct this path from $a_0$ towards $a_m$ and call this the direction of ``descent'' since $\ell(\cdot)$ decreases along the path. Note that any $\text{sn\textcolor{black}{-path}}$ consists of at most one point from each onion layer and precisely one point from the first onion layer.

Let $(p,q)$ be some edge in $T$ and suppose that $\text{sn\textcolor{black}{-path}}(p)$ ends at $p^{\prime}\in\setp^{(1)}$ and $\text{sn\textcolor{black}{-path}}(q)$ ends in $q^{\prime} \in \setp^{(1)}$. There are two paths in $T$ from $p^\prime$ to $q^\prime$ along $\Conv(\setp)$, one in the clockwise direction and the other in the counter-clockwise direction. Each of these paths along with the edge $(p,q)$ and the two $\text{sn\textcolor{black}{-paths}}$ starting at $p$ and $q$ respectively, defines a region within $\Conv(\setp)$. We call these two regions the $\text{sn}$-regions of $(p,q)$. See Figure~\ref{c-tri:sections:t-paths:figs:21}. Given any $\text{sn}$-region $R$, we refer to the number of triangles in any triangulation of $R$ as the \emph{size} of $R$. This is well defined since the number of triangles is the same regardless of the triangulation chosen.

\begin{figure}[!htb]
	\begin{center}
		\begin{minipage}[b][6.5cm][t]{7cm}
			\begin{center}
				\includegraphics[height=4.2cm]{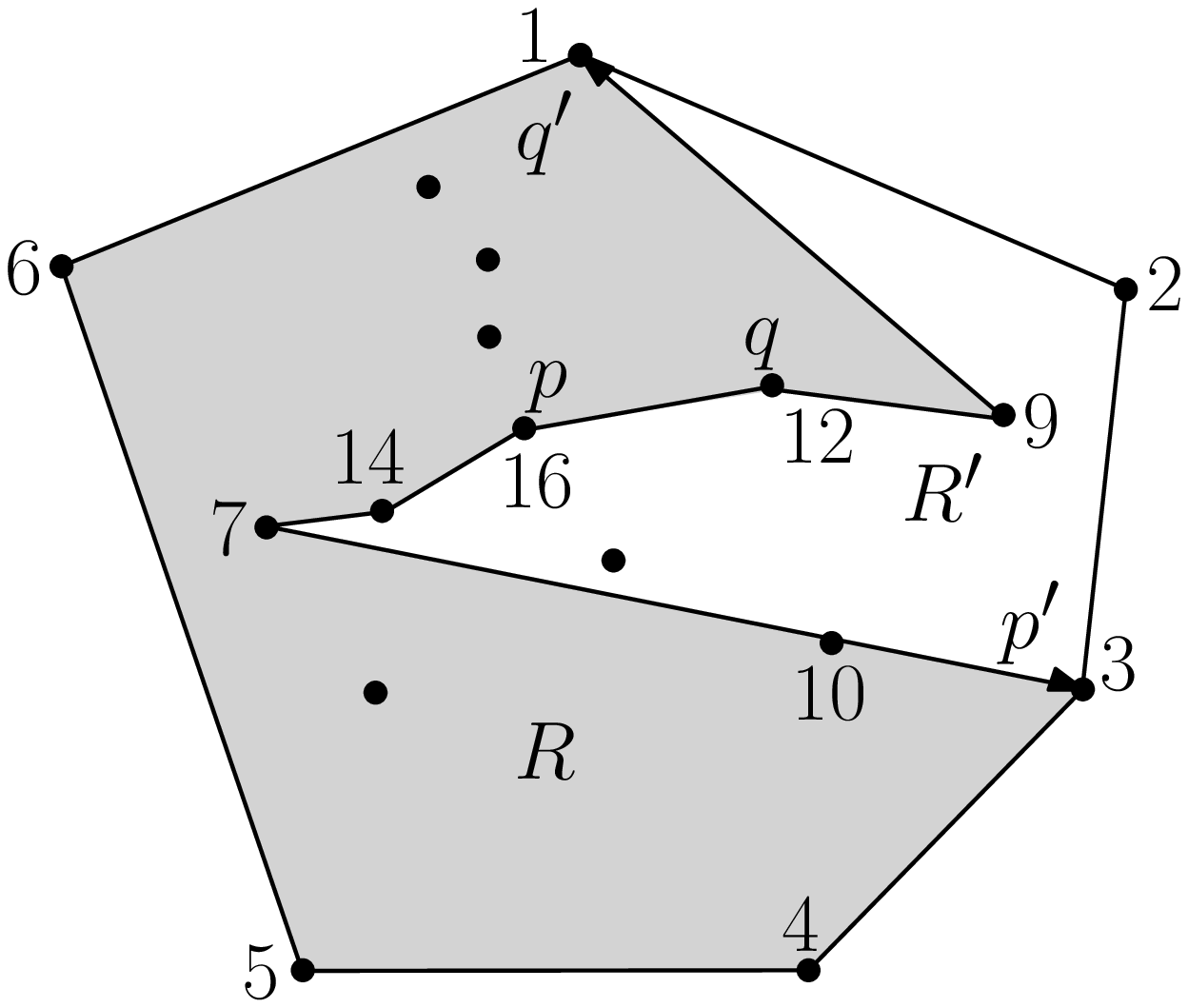}
				\caption{$R$ and $R^\prime$ are the $\text{sn}$-regions of $(p,q)$.}
				\label{c-tri:sections:t-paths:figs:21}
			\end{center}
		\end{minipage}
	\quad
		\begin{minipage}[b][6.5cm][t]{7cm}
			\begin{center}
				\includegraphics[height=4.2cm]{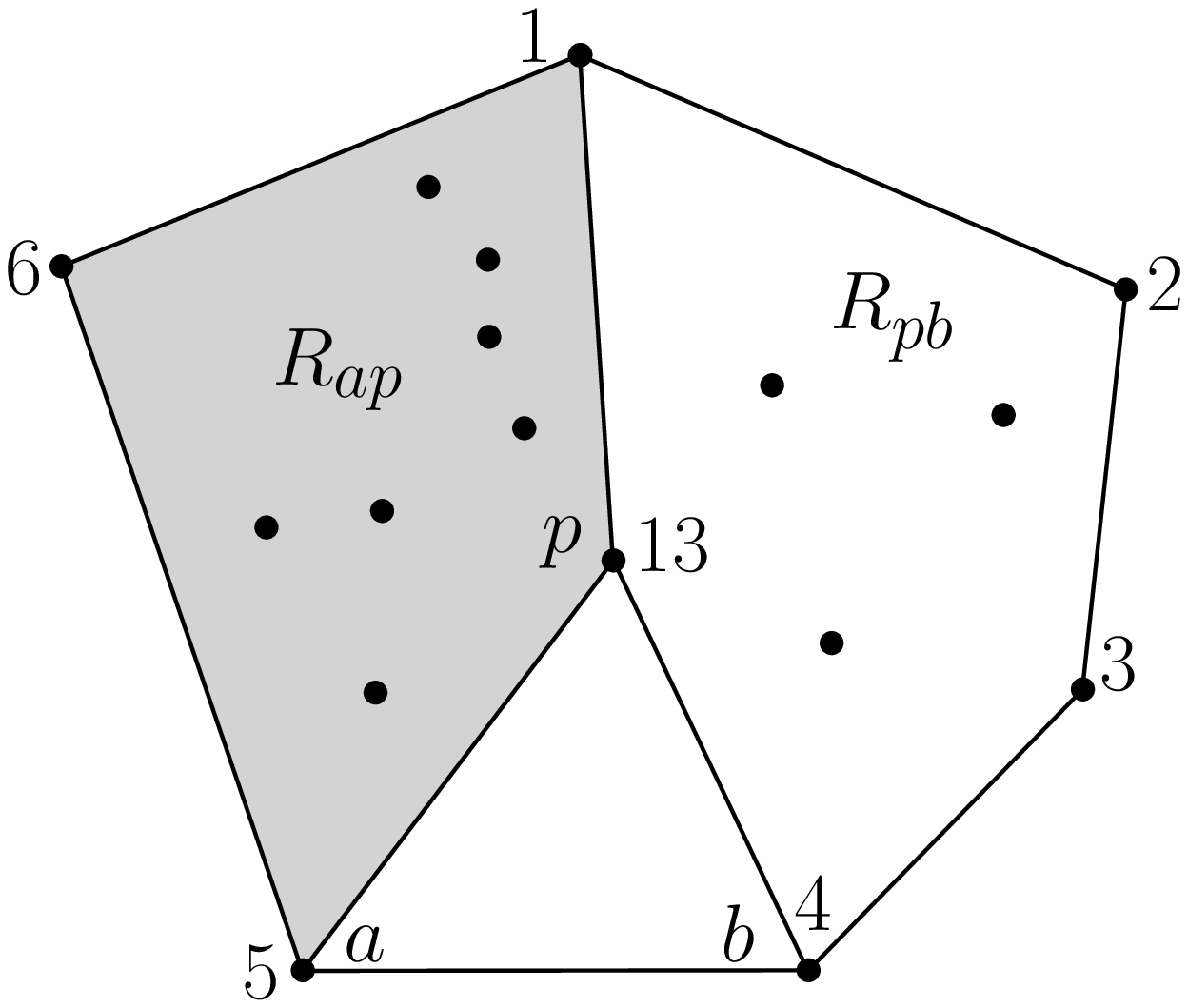}
				\caption{$R_{ap}$ and $R_{pb}$ are the $\text{sn}$-regions of $(a,p)$ and $(p,b)$, respectively, that do not contain triangle $apb$.}
				\label{c-tri:sections:t-paths:figs:22}
			\end{center}
		\end{minipage}
	\end{center}
\end{figure}

Let $ab$ be an edge on $\Conv(\setp)$. Observe that in any triangulation, $\Conv(\setp)$ is one of the $\text{sn}$-regions of $(a,b)$, the other region being empty. In any triangulation $T$ of $\setp$, there is precisely one triangle $apb$ that the edge $ab$ belongs to. Let $R_{ap}$ be the $\text{sn}$-region of $(a,p)$ that does not contain $apb$ and similarly let $R_{pb}$ be the $\text{sn}$-region of $(p,b)$ that does not contain $apb$, see Figure~\ref{c-tri:sections:t-paths:figs:22}.

\subsection{The algorithm}\label{c-tri:sections:triple-paths:algorithm}

Let $ab$ be again an edge on $\Conv(\setp)$. The core idea of our algorithm is as follows: We can easily enumerate all the points $p$ such that the triangle $apb$ appears in some triangulation. This is just the set $Q$ of points $p$ such that the triangle $apb$ is free of other points of $\setp$. For every element $p$ of $Q$, suppose that we can enumerate the $\text{sn}$-paths $\rho$ of $p$ over all triangulations of $\setp$. For ever pair $(p,\rho)$, let $\T_{(p,\rho)} = \T_{(p,\rho)}(\setp)$ be the set of triangulations of $\setp$ that contain the triangle $apb$ and in which $\rho$ is the $\text{sn}$-path of $p$. If, for each such pair that we can obtain, we can compute $|\T_{(p,\rho)}|$, then we are done, since each triangulation of $\setp$ must contain precisely one pair $(p,\rho)$, adding the numbers over all pairs gives us the total number of triangulations. 

Let us fix a pair $(p,\rho)$ for which we would like to compute $|\T_{(p,\rho)}|$. The pair already defines the regions $R_{ap}$ and $R_{pb}$ for all triangulations in $\T_{(p,\rho)}$. Observe that any triangulation in $\T_{(p,\rho)}$ contains a triangulation $T_{ap}$ of $R_{ap}$ and a triangulation $T_{pb}$ of $R_{pb}$, each of which satisfy the following $\text{sn}$-constraint: For each edge $(q,r)$ in $\rho$ there is no edge $(q,s)$ in the triangulation (either $T_{ap}$ or $T_{pb}$) such that $s<r$. Furthermore, putting together any pair of triangulations $T_{ap}$ and $T_{pb}$, each satisfying the constraint, and the triangle $apb$ gives a triangulation in $\T_{(p,\rho)}$. This observation follows from the fact that $\rho$ is an $\text{sn}$-path of $p$ in any triangulation of  $\T_{p,\rho}$, and allows us to separately compute the number of ($\text{sn}$-constraint-satisfying) triangulations $N_{ap}$ of $R_{ap}$ and $N_{pb}$ of $R_{pb}$ whose product gives $|\T_{(p,\rho)}|$.

The numbers $N_{ap}$ and $N_{pb}$ are computed recursively. We will maintain the invariant that at any point in the recursion we are dealing with an $\text{sn}$-region of some edge. This is certainly true in the beginning since we start with an $\text{sn}$-region of the edge $ab$ and also in the next step since we recurse on $\text{sn}$-regions defined by the edges $(a,p)$ and $(p,b)$ respectively. At any point, let us say that we are dealing with an $\text{sn}$-region $R$ defined by an edge $(x,y)$ and let $\rho_x$ and $\rho_y$ be the $\text{sn}$-paths starting at $x$ and $y$ respectively.

Now, we recurse almost exactly as we did before: We enumerate the set of points $z$ such that the triangle $xzy$ lies within $R$ and is free of other points of $\setp$ contained in $R$, see Figure~\ref{c-tri:sections:t-paths:figs:23}. Furthermore, we ensure that if $z$ happens to be a point in either $\rho_x$ or $\rho_y$, and $(z,t)$ is an edge in that $\text{sn}$-path, then both $x$ and $y$ are bigger than $t$. This way, we do not violate the $\text{sn}$-constraint. For each such $z$ we enumerate the portions of $\text{sn}$-paths starting at $z$ that lie within $R$. See Figure~\ref{c-tri:sections:t-paths:figs:23}. Each such path splits the region $R$ into regions $R_{xz}$ and $R_{zy}$ which are $\text{sn}$-regions defined by $(x,z)$ and $(z,y)$ respectively. Each of the regions $R_{xz}$ and $R_{zy}$ have sizes smaller than $R$, \emph{i.e.}, fewer triangles in any triangulation. The recursion bottoms out when the size is $\leq 1$, in which case we know that there is exactly one triangulation. Note that even though we enumerate only the portions of the $\text{sn}$-paths of $z$ that lie within $R$, these portions implicitly define the entire $\text{sn}$-path of $z$. This is because such a portion either ends at a point on the first onion layer in which case it is the entire $\text{sn}$-path, or at a point $w$ on either $\rho_x$ or $\rho_y$. The direction of descent along that $\text{sn}$-path, starting at $w$, is then the remaining portion of the $\text{sn}$-path of $z$.

\begin{figure}[!hbt]
    \begin{center}
	\includegraphics[height=4.2cm]{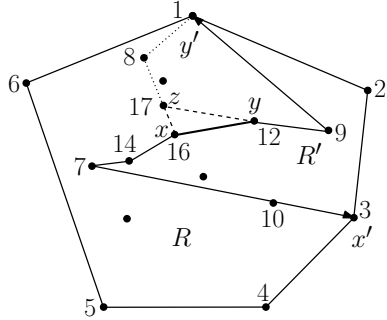}
	\caption{$R$ and $R^\prime$ are the $\text{sn}$-regions of $(x,y)$.}
	\label{c-tri:sections:t-paths:figs:23}
	\end{center}
\end{figure}

One detail is still missing. How do we enumerate the portions of the $\text{sn}$-paths of $z$ that belong to at least one triangulation of $R$? and the answer is: We will not do it. Instead, we enumerate a superset of paths which are \emph{descending} in the sense that they start at $z$ and each successive point is in a strictly upper layer (a layer containing points with smaller indices). Again, we only enumerate the portion of such paths that lie inside $R$ since the rest is implicitly defined. For any descending path that does not really belong to any triangulation of $R$, at least one of the regions $R_{xz}$ or $R_{zy}$ has no triangulations satisfying the $\text{sn}$-constraint. This will be detected somewhere down the recursion where we will not be able find any $z$ satisfying the $\text{sn}$-constraint. At that point, we return $0$ as the number of triangulations. Thus the algorithm still works in these cases. 

There is one other ingredient that we add for efficiency: Memoization. Whenever we compute the number of triangulations of a certain $\text{sn}$-region that satisfy the $\text{sn}$-constraint dictated by the $\text{sn}$-paths defining the region, we store it in a hash table (or any other data structure). Now consider a \emph{call graph} in which each node represents an $\text{sn}$-region and there is a directed edge from a region $R$ to a region $R^\prime$ if from $R$ we make a recursive call to $R^\prime$. The number of egdes in this graph is an upper bound\footnote{Up to a polynomial overhead arising from the construction and handling of sub-problems.} on the running time of the algorithm since, because of memoization, no edge is \emph{traversed} more than once. 

We will now prove an upper bound on the number of edges in the call graph. Each call from a region $R$ to a region $R^\prime$ can be charged to a triple of descending paths - two defining $R$ and a third that, along with a triangle, splits $R$ into two regions, one of which is $R^\prime$. The triples $(\rho_1, \rho_2, \rho_3)$ that are produced in the algorithm have the property that once two paths merge in the direction of descent, they never split again. This is ensured by the fact that we only enumerate the portions of the third descending path within the region $R$ and the rest is implicitly defined, as noted before. Let $\rho^\prime_2$ be the portion of $\rho_2$ that does not have any point in common with $\rho_1$, and let $\rho^\prime_3$ be the portion of $\rho_3$ that does not have any point in common with either $\rho_1$ or $\rho_2$. The descending paths $\rho_1$, $\rho_2^{\prime}$ and $\rho_3^{\prime}$ are vertex disjoint, and along with some additional information they completely describe $\rho_1$, $\rho_2$ and $\rho_3$. The additional information that is required is whether, and where, $\rho_2$ merges with $\rho_1$, and whether, and where, $\rho_3$ merges with one of the other paths. If $\setp$ has $k$ onion layers, then each descending path has length at most $k$ and therefore there are at most $k$ ways that $\rho^\prime_2$ may merge with $\rho_1$, and at most $2k$ ways $\rho^\prime_3$ may merge with one of $\rho_1$ or $\rho_2$. Therefore, if $U$ is an upper bound on the number of triples of vertex disjoint descending paths, then $2k^2U$ is an upper bound on the number of triples $(\rho_1,\rho_2,\rho_3)$ as described above, and hence also an upper bound on the running time of the algorithm. 

\subsection{Number of vertex-disjoint triples of descending paths}\label{c-tri:sections:triple-paths:num-paths}

Each descending path uses at most one vertex from every onion layer.  Let $n_i = |\setp^{(i)}|$ be the size of the $i$-th onion layer. Let us count how many ways there are for any triple of paths to use at most one vertex each from this layer. There is one way for the triple of paths to skip this onion layer. There are $n_i$ ways of choosing one point among the $n_i$ which may then be used by any of the paths. This gives $3n_i$ ways for the three paths. There are $\binom{n_i}{2}$ ways to choose two points, and any two of the paths may use them. This gives $6\binom{n_i}{2}$ ways among the three paths. Finally there are $\binom{n_i}{3}$ ways of choosing three points, and there are three (not six) ways for the three paths to use one of these vertices. This is because these paths are non-crossing planar curves, and therefore the clockwise order of these paths along any $\Conv\left(\setp^{(i)}\right)$ that intersects all three of them is the same for each $i$. The overall number of ways in which at most three points can be used from the $i$-th layer is therefore $f(n_i)$, where $f(x) = 1 + 3x + 6\frac{x(x-1)}{2} + 3\frac{x(x-1)(x-2)}{6}$. 

The number of triples of vertex disjoint descending paths is therefore at most $U=\prod_{i=1}^{k}f(n_i)$. Since each $n_i$ is a positive integer, and the function $f(\cdot)$ is log-concave for $x\geq 1$%\footnote{$f$ is log-concave iff $f\left(\frac{x+y}{2}\right)^{2}\geq f(x)\cdot f(y)$}
, the above product is maximized when each $n_i$ is equal to $\frac{n}{k}$. This gives an upper bound of $f\left(\frac{n}{k}\right)^k = g\left(\frac{n}{k}\right)^n$, where $g(x) = f(x)^{\frac{1}{x}}$. Now, $g(x)$ is maximized for some value of $x$ between $0$ and $1$ and is a decreasing function for $x\geq 1$. Since each onion layer except the $k$-th one must have at least three points, we have $U=O\left(g(3)^n\right)$. The fact that the $k$-th onion layer may have fewer than three points makes only a difference of a constant factor. Therefore the running time of the algorithm presented in this section is $O\left(k^2g(3)^n\right) = O^{*}(3.1414^n)$. This concludes the proof of Theorem~\ref{c-tri:theorems:triple-paths}.

We want to point out that often the number of onion layers can be much smaller than the maximum possible $\left\lceil\frac{n}{3}\right\rceil$. For example, Dalal~\cite{DBLP:journals/rsa/Dalal04} has shown that if $n$ points are chosen uniformly at random from a disk, then the expected number of onion layers of the resulting point set is $\Theta\left(n^{2/3}\right)$.

From a theoretical point of view, the algorithm presented in this section, sn-path algorithm for short, has a running time polynomial in $n$ whenever the number of onion layers of $\setp$ is constant. This is the first known algorithm for counting triangulations having this property. Also, its worst-case behavior is better than the one based on T-paths presented in~\cite{sweep-line}, $O^{*}(3.1414^{n})$ of the former against $O^{*}(9^{n})$ of the latter. 

The sn-path algorithm is an excellent candidate for having good experimental behavior as well, due to its polynomial-time instances. Towards the end of paper we shall see how the sn-path performs experimentally against the fastest-known (in practice) algorithm for counting triangulations presented in~\cite{ray-seidel}.

\section{Counting other crossing-free structures}\label{c-tri:sections:triangular-paths}

In this section we show how the ideas of the sn-path algorithm can be ``modified'' or ``adapted''  so we can develop a general framework that helps to count crossing-free structures in general. We use this framework to count perfect matchings and spanning cycles defined by $\setp$. 

\subsection{Counting matchings and spanning cycles}\label{c-tri:sections:other-cfs:mandc}

Assume we want to count crossing-free matchings spanned by $\setp$. Clearly any matching can be completed to a triangulation by adding edges, and thus we might want to try the technique used for counting triangulations: Take a set $\SS$ of separators and for each $S \in \SS$ count the matchings in triangulations containing $S$, and finally add this up over all $S \in \SS$. In any matching $M$ that can be completed to a triangulation containing $S$, each vertex in $S$ is either unmatched, or it is matched to a vertex within some $R^S_i$, or it is matched to another vertex in $S$. We can \emph{annotate} each separator $S$ with this information. When counting, for each $S\in\SS$, we iterate over all annotations of $S$, and take care to be consistent with the current annotation when recursing into the sub-problems.

This simple algorithm fails because some matchings $M$ could be contained in a triangulation that could contain several, say $s_M>1$, separators and would thus fool our algorithm to count $M$ exactly $s_M$ times. If $s_{M} = s$ were a constant over all matchings we would not have this problem, however, we are not aware of any set of separators $\SS$ with this property. 

There is however a way in which we can modify the simple algorithm so that we can count each matching exactly once: We embed each matching $M$ into a unique triangulation $T \supset M$. Given a family $\SS$ of separators for the triangulations of $\setp$, we associate a unique $S \in \SS$ to each matching. For concreteness, let us associate to each $M$ the constrained Delaunay triangulation (CDT) $\triangle^M$ \emph{constrained} to contain $M$, which we briefly describe next.

\paragraph{Constrained Delaunay Triangulation:}

The constrained Delaunay triangulation (CDT) $\triangle^{S}$ of $\setp$ was first introduced in~\cite{DBLP:conf/compgeom/Chew87}. Formally, it is the triangulation $T$ of $\setp$ containing $S$ such that no edge $e$ in $T\setminus S$ is flippable in the following sense: Let $\triangle_{1}, \triangle_{2}$ be triangles of $\setp$ sharing $e$. The edge $e$ is flippable if and only if $\square = \triangle_{1}\cup\triangle_{2}$ is convex, and replacing $e$ with the other diagonal of $\square$ increases the smallest angle of the triangulation of $\square$. One of the most important properties of constrained Delaunay triangulations is its uniqueness if no four points of $\setp$ are cocircular. Thus, under standard non-degeneracy assumptions, there is a unique CDT for any given set of mandatory edges. For a good study on constrained Delaunay triangulations we suggest the book~\cite{Hjelle:2006:TA:1214284} by {\O}. Hjelle and M. D{\ae}hlen.

For our counting purposes we will assume that no four points of $\setp$ are cocircular. This can easily be taken care of by perturbing $\setp$. We can now go back to our simple algorithm for counting matchings and revise it as follows: Whenever we recurse, in each sub-problem we only count matchings $M$ such that $S \subseteq\triangle^M$, where $S\in\SS$ is a separator. If this last condition can be satisfied locally in each sub-problem, \emph{i.e.}, choices in one sub-problem do not depend on choices in others, we are done. While not every $\SS$ admits such a locality condition, some do as we will see next.

\subsection{Triangular paths}

We assume again that $\setp$ has $k$ onion layers. For every point~$p \in\setp$ (on layer $\setp^{(i)}$ which is not the first layer) we fix in advance a ray $\rho_p$ which emanates from $p$ and does not intersect the interior of $\Conv\left(\setp^{(i)}\right)$.

For any triangulation $T$ of $\setp$ there is a unique triangle $\triangle_{p} = p,q_1,q_2$ adjacent to $p$ and intersecting $\rho_p$. Let $q_p$ be the smaller of $q_1$ and $q_2$, using the same labeling as before. Clearly $q_p$ lies in a layer lower than the one containing $p$. Let $p_0,p_1,\ldots,p_r$ be the sequence so that $p_0=p$, $p_{i+1} = q_{p_i}$, $\forall\ 0\leq i < k$, and $p_r$ lies on the first layer. We call $\setp_p(T) := \bigcup_i \triangle_{p_i}$ the \emph{triangular path} of $p$ w.r.t.~$T$, and we call $p_r$ the \emph{last point} of $\setp_p(T)$. See Figure~\ref{c-tri:sections:t-paths:figs:24}.

\begin{figure}[!hbt]
    \begin{center}
	\includegraphics[height=4cm]{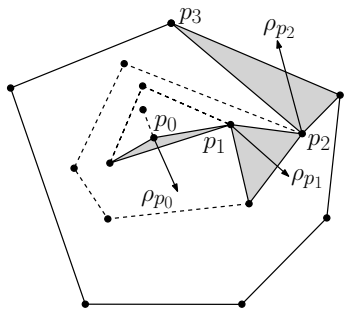}
	\caption{Triangular path $\setp_{p}$ starting in onion layer $\setp^{(4)}$. Onion layers are shown in dashed. $\setp_{p}$ can be extended to a triangulation $T$, in such a case $\setp_{p}$ will be unique for $T$.}
	\label{c-tri:sections:t-paths:figs:24}
	\end{center}
\end{figure}

The triangular path $\setp_p(T)$ is uniquely defined for any triangulation. Moreover, for distinct triangulations $T_1$ and $T_2$, $\setp_p(T_1), \setp_p(T_2)$ are either identical or they intersect properly: Let $i$ be the first position where $\triangle_{p_i}(T_1) \ne \triangle_{p_i}(T_2)$, then those two triangles intersect, as they both are adjacent to $p$, intersect $\rho_p$ and have interiors free of points in $\setp$. We are now ready to finish the algorithm for counting matchings.

\subsubsection{Algorithm for counting matchings}

Given a matching $M$, let $\triangle^{M}$ be the CDT of $M$. By our assumption of no four cocircular points, this CDT is unique for $M$. We annotate $\triangle^{M}$ as follows:
\begin{itemize}
	\item each vertex $v$ of $\triangle^{M}$ is annotated with a number $m_{v}$ that represents the vertex of $M$ that $v$ is matched to. If $m_{v}$ is $0$ say, then we know that $v$ is not matched in $M$.
	\item each edge $e$ of $\triangle^{M}$ is annotated with a bit $b_{e}$ that indicates whether $e$ belongs to $M$ or not.
\end{itemize}

Let us denote by $\CDTA^{M}$ the annotated version of $\triangle^{M}$. Let $S$ be a separator contained in $\triangle^{M}$ that splits $\Conv(\setp)$ into regions $R_{1},\ldots, R_{t}$. Separator $S$ inherits all the information from $\CDTA^{M}$. We additionally keep track of whether $m_{v}$ is any of the adjacencies of $v$ in $S$, for each vertex $v\in S$. If not then we set $m_{v}$ to the index $1\leq i\leq t$ of the region $R_{i}$ the matching vertex of $v$ falls into (it must necessarily be a region having $v$ as a vertex of its boundary). The separator thus annotated will be denoted by $\CDTA^{M}_{S}$.

We say that an annotated constrained Delaunay triangulation is \emph{legal} if and only if it is identical to $\CDTA^{M}$, for some matching $M$. Since there is a one-to-one correspondence between matchings and legal constrained Delaunay triangulations, our goal is to count the latter.

Our algorithm is essentially the same as for counting triangulations: Instead of $\text{sn}$-paths we use annotated triangular paths. We start with an edge $ab$ on $\Conv(\setp)$, and enumerate the set of points $p$ such that the triangle $apb$ is free of other points of $\setp$. For each such $p$, the triangle $apb$ along with the triangular path starting at $p$ forms a separator, see Figure~\ref{c-tri:sections:t-paths:figs:25}. We enumerate such separators and all possible annotations for each one of them. Each such annotated separator splits $\Conv(\setp)$ into two smaller regions in which we recurse. In each recursive sub-problem we count (legal) annotated constrained Delaunay triangulations consistent with the annotated separator.

\begin{figure}[!hbt]
    \begin{center}
	\includegraphics[height=4cm]{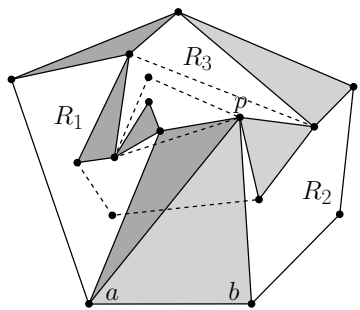}
	\caption{In the first call of the algorithm, the triangular path shown in dark gray is created. It divides the problem into regions $R_{1}\cup R_{3}$ and $R_{2}$. A call for the latter creates the triangular path shown in light gray. Annotations are not shown for simplicity.}
	\label{c-tri:sections:t-paths:figs:25}
	\end{center}
\end{figure}

The reason for which we use triangular paths instead of simple $\text{sn}$-paths is the following: No edge in a separator, formed by a triangular path, lies on the boundary of more than one sub-problem. This allows us to verify flippability of edges separately in each sub-problem. If an edge belonged to more than one sub-problem, then the flippability of this edge would depend on the choices made in each sub-problem, thus introducing dependency between these sub-problems.

As in the case for counting triangulations, we use \emph{memoization}. The running time as before is dominated by the number of triples of annotated triangular paths. The size of each triangular path is $O(k)$, thus there are clearly at most $n^{O(k)}$ triangular paths. Also, trying all possible annotations per triangular path leads to no more than $n^{O(k)}$ annotations per triangular path, as can be easily checked. Hence there are $n^{O(k)}$ annotated triangular paths, and also $n^{O(k)}$ triples of annotated triangular paths. The overall running time is thus $n^{O(k)}$, which even considers the polynomial overheads arising from checking flippability of edges and inclusion of points into sub-problems. This concludes one part of Theorem~\ref{c-tri:theorems:triangular-paths}.

The annotations required for counting matchings are not very complicated, but for many other counting problems this is a highly non-trivial task. An example of more involved annotations is given next, where we consider the problem of counting spanning cycles.

\subsubsection{Algorithm for counting spanning cycles}

Counting spanning cycles is more complicated than counting matchings. What we will actually do is that, instead of counting spanning cycles, we will count \emph{rooted and oriented} spanning cycles. Given any cycle, we make it rooted by designating a \emph{starting vertex}, and we make it oriented by assigning an \emph{orientation}- clockwise or counter-clockwise. We then number the vertices in the cycle from $1$ to $n$, beginning at the starting vertex (which is the root of the cycle), and continuing along the assigned direction. We also direct the edges along this direction. This way, each spanning cycle is counted exactly $2n$ times. At the end we divide the computed number by $2n$ to get the desired number. In the remainder we use the term HamCycle for rooted and oriented spanning cycles.

\newpage

Given a HamCycle $H$ let $\triangle^{H}$ be the CDT of $H$. We annotate $\triangle^{H}$ as follows:
\begin{itemize}
	\item each vertex $v$ of $\triangle^{H}$ is annotated with $(\text{pos}_{v}, \text{prev}_{v}, \text{next}_{v})$, where $\text{pos}_{v}$ is the number assigned to $v$ in $H$, $\text{prev}_{v}$ is the vertex lying immediately before $v$ in $H$, and $\text{next}_{v}$ is the vertex lying immediately after $v$ in $H$.
	\item each edge $e$ in $\triangle^{H}$ is annotated with a bit $b_{e}$ that indicates whether $e$ belongs to $H$ or not.
\end{itemize}

As in the case for matchings, the annotated $\triangle^{H}$ will be denoted by $\CDTA^{H}$. Let $S$ be again a separator contained in $\triangle^{H}$ that splits $\Conv(\setp)$ into regions $R_{1},\ldots R_{t}$. Separator $S$ inherits the following information from $\CDTA^{H}$: Each vertex $v\in S$ inherits $\text{pos}_{v}$ from $\CDTA^{H}$. If $\text{prev}_{v}$ and $\text{next}_{v}$ are already adjacent to $v$ in $S$ then this information is also inherited. If $\text{prev}_{v}$ is absent in $S$ then $v$ is annotated with the index $i$, $1\leq i\leq t$, of the region $R_{i}$ that $\text{prev}_{v}$ falls in. The same holds for $\text{next}_{v}$. Each edge $e$ of $S$ carries the annotation it has in $\CDTA^{H}$. The separator $S$ of $\Delta^{H}$ thus annotated will be denoted by $\CDTA^{H}_{S}$.

The algorithm, as the reader might be thinking right now, is no other than the algorithm for counting matchings. The only difference are the annotations, they encode a different problem. Thus again, the number of triangular paths is $n^{O(k)}$. The number of annotations per triangular path stays $n^{O(k)}$, and hence the total running time will stay at $n^{O(k)}$, including again the other polynomial overheads. This finishes the proof of Theorem~\ref{c-tri:theorems:triangular-paths}.

\section{The hardness result}\label{c-tri:sections:hardness}

In this section we show a hardness result related to our counting algorithms. Observe that those algorithms are parameterized by the number $k$ of onion layers of $\setp$. They have running times of the sort $n^{O(k)}$. Thus, from the complexity point of view, it is natural to ask whether algorithms with running times of the sort $g(k)\cdot n^{O(1)}$, for some function $g$ independent of $n$, are possible. That would mean that our problems belong to the FPT complexity class. Unfortunately, our techniques are general enough to solve harder problems, such as the {\sc Restricted-Triangulation-Counting-Problem} explained before, on page~\pageref{c-tri:theorems:fpt}.

Here we prove Theorem~\ref{c-tri:theorems:fpt}, which states that the {\sc Restricted-Triangulation-Counting-Problem}, RTCP for short, is W[2]-hard if the parameter is considered to be the number of onion layers of $\setp$. More, this result even holds for the problem of just deciding the \emph{existence} of a restricted triangulation. 

The algorithms of Theorems~\ref{c-tri:theorems:triple-paths} and~\ref{c-tri:theorems:triangular-paths} require little to no modification to be run on instances of RTCP, that is, those algorithms are quite generic. Since the separation FPT $\neq$ W[2] is widely believed, and we do not really know about the complexity of the counting problems studied in this paper, we can still hope that, by exploiting structural properties, we could obtain fixed-parameter tractable algorithms for them. The book by J. Flum and M. Grohe, see~\cite{flum-grohe}, is an excellent reference for Parameterized Complexity Theory.

\subsection{Preliminaries}\label{secs:hardness}

Let $\setp$ be a set of $n$ points with $k$ onion layers, and let $E$ be some set of pre-specified edges spanned by $\setp$. We say that a triangulation $T$ of $\setp$ is \emph{restricted} w.r.t.~$E$ if $T \subseteq E$. Here we consider the following {\sc Restricted-Triangulation-Existence-Problem}: On input $(\setp,E)$, decide whether there exists a triangulation of $\setp$ that is restricted w.r.t.~$E$. This defines the {\sc Restricted-Triangulation-Counting-Problem} in the natural way, and the existence problem is trivially reducible to the counting problem.

The {\sc Restricted-Triangulation-Existence-Problem}, RTEP for short, was\\ proven to be NP-complete in~\cite{DBLP:conf/focs/Lloyd77,schulz}. Something very important can be observed here, namely, both reductions are actually parsimonious\footnote{This means that there is a one-to-one correspondence between the solution sets.}, implying \#P-completeness of its natural \emph{counting} problem, RTCP.

So far all reductions involving restricted triangulations rely heavily on the ability to specify a particular set $E$ as part of the input. If $E$ is instead fixed to the set of \emph{all} edges spanned by $\setp$, we obtain the problem of counting \emph{all} triangulations of $\setp$, which we strongly believe to be \#P-complete. 

In this section we parameterize RTCP and RTEP by $k$, the number of onion layers of $\setp$. As we mentioned before, the counting algorithm, for triangulations, presented in Section~\ref{c-tri:sections:triple-paths} can easily be adapted to solve RTCP, and thus also RTEP, in time $n^{O(k)}$. Our proof is by reduction from the {\sc Parameterized-Hitting-Set-Problem}, PHSP for short, which is proven to be W[2]-hard in~\cite{flum-grohe}. An instance $A$ of this problem is formed by numbers $n,m,k\in\mathbb{N}$, along with sets $S_{1},\ldots,S_{m}\subseteq[n]$, where $k$ is considered a parameter, and $[n]:=\{0,\ldots,n-1\}$. The output to $A$ is ``yes'' iff there is a set $H\subseteq[n]$ of size at most $k$, such that $H\cap S_{i}\neq\emptyset$ for every $1\leq i\leq m$. 

In our reduction, several gadgets are used to transform an instance $A$ of the hitting set problem to an instance $G_A=(\setp,E)$ of the {\sc Restricted-Triangulation-Existence-Problem}. The reduction is an \emph{fpt-reduction} in the sense of \cite{flum-grohe}, that is, it maps every instance $A$ with parameter $k$ to an instance $G_A$ with $O(k)$ onion layers. Each gadget is given by a set of points with $O(1)$ onion layers, along with a set of pre-specified edges. The gadgets that will be used are called: \emph{pipes}, \emph{wires}, \emph{ORs}, \emph{terminals}, \emph{testers}, and \emph{crossings}, their specifications will be given later on, for now we would like to explain how the gadgets fit together as well as the intuition behind it.

\subsection{Construction and intuition}

Given an instance $A$ of PHSP, as explained above, we will create in polynomial time an instance $G_{A}=(\setp,E)$ of RTEP of size $\textup{poly}(n,m)$ that has $O(k)$ onion layers and admits a triangulation w.r.t.~$E$ iff $A$ admits a hitting set of size $\leq k$. The mapping $A \mapsto G_A$ will clearly be polynomial-time computable, and thus an fpt-reduction. Figure~\ref{c-tri:sections:t-paths:figs:29} is a reference for the construction that follows.

In the construction, we start with parallel pipes $Q_{1},\ldots,Q_{k}$ of $n$ states each, and of length polynomial in $m$ and $n$. Pipe $Q_{i}$ lies above pipe $Q_{i+1}$. Let $Q_{i}$ be a pipe, $1\leq i\leq k$, and let $S_{j} = \{s_{j,1},\ldots s_{j,t}\}\subseteq[n]$ be a set of instance $A$. 
We define the \emph{stripe} $B_{i,j}$ as a set of $t$ testers attached to $Q_{i}$ that check if $Q_{i}$ carries any of the values of set $S_{j}$, see Figure~\ref{c-tri:sections:t-paths:figs:29}. The stripe $B_{i+1,j}$ will lie in the same vertical slab as $B_{i,j}$. The testers of $B_{i,j}$ are connected to a chain of or-gadgets that lies between pipes $Q_{i}, Q_{i+1}$. For $i < k$, the output of the last or-gadget in $B_{i,j}$ is carried to $B_{i+1,j}$ by a crossing-gadget, see Figure~\ref{c-tri:sections:t-paths:figs:29}. For $i=k$, the last or-gadget in $B_{i,j}$ is connected to a terminal-gadget.

\begin{figure}[!htb]
	\begin{center}
	\includegraphics[scale=0.8]{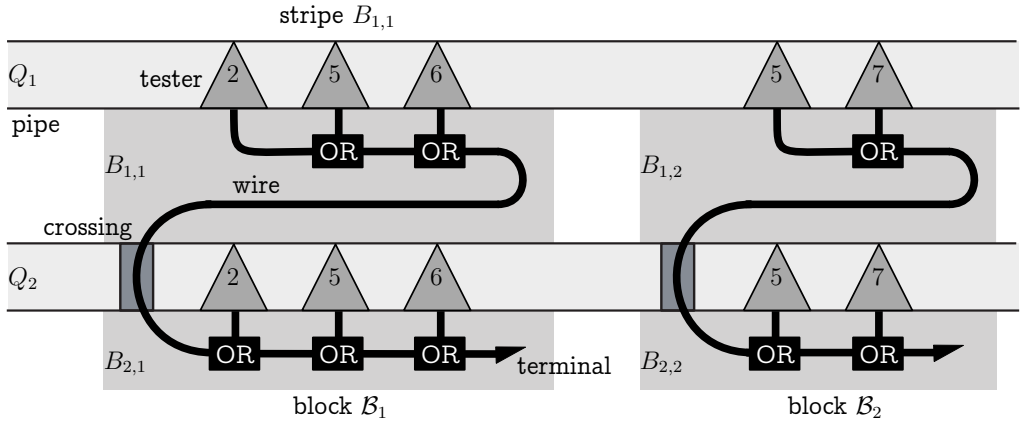}
	\end{center}
	\caption{Instance $G_{A}$ produced from instance $A$ of the {\sc Parameterized-Hitting-Set-Problem} with $n=8, m=2, k=2$ and $S_{1}=\{2,5,6\}$, $S_{2}=\{5,7\}$.}
	\label{c-tri:sections:t-paths:figs:29}
\end{figure}

The \emph{block} $\mathcal{B}_{j}$ is the union of the stripes $B_{1,j},\ldots, B_{k,j}$. The blocks $\mathcal{B}_{1},\ldots,\mathcal{B}_{m}$ are arranged horizontally in such a way that the points in stripes $B_{i,1},\ldots,B_{i,m}$, with $1\leq i\leq k$, are horizontally collinear, that is, they are aligned by their $y$-coordinate.

Finally, $\setp$ is defined to be the set of points of all the gadgets involved. To define the set $E$ of pre-specified edges, we first include the edges of all gadgets involved. Then, the empty spaces between gadgets are triangulated arbitrarily, and these edges are added to $E$. We now set $G_A = (\setp,E)$.

The intuition behind the construction is the following: Horizontally, pipe $Q_{i}$ transmits a single value between $1$ and $n$. The testers in stripe $B_{i,j}$ verify if the value transmitted by $Q_{i}$ hits one of the elements of the set $S_{j}$ of $A$. If so, this information is transmitted vertically along block $\mathcal{B}_{j}$, in such a case the transmitted value is $\textup{true}$. For this transmission we need ORs, wires and crossing gadgets. At the end of block $B_{j}$ the terminal gadget can be triangulated iff the value transmitted to it is $\textup{true}$. If $S_{j}$ is not hit by the value transmitted in $Q_{i}$, then the testers will transmit $\textup{false}$ and this value will be transmitted vertically along $\mathcal{B}_{j}$ until it is possibly flipped by another pipe $Q_{r}$, with $i < r\leq k$, thus $S_{j}$ is not hit by $Q_{i}$ but it is hit by $Q_{r}$. If the value transmitted to the terminal gadget in block $\mathcal{B}_{j}$ is $\textup{false}$, this means that the terminal cannot be triangulated, thus no restricted triangulation of $G_{A}$ exists. This in turns implies that $S_{j}$ was not hit by any value transmitted by the pipes $Q_{1},\ldots, Q_{k}$. If this is always the case then no hitting set of size at most $k$ exists for $A$. 

All this will be formally proven later on, for now we believe that this rough intuition is enough. Therefore we will jump now to define the gadgets formally.

\subsection{Defining the gadgets}

The basic gadget is the \emph{pipe}, shown in Figure~\ref{c-tri:sections:t-paths:figs:26}, whose definition is the following:

\begin{definition}[Pipe]
A \emph{pipe} $Q$ with $n$ states and length $l > 4(n-1)$ consists of points $p_{1}\ldots p_{l}$, $q_{1}\ldots q_{l}$ with $p_{t}=(t,0)$, $q_{t}=(t,1)$, $1\leq t\leq l$, and a set $E_{Q} = S\cup F\cup L_{0}\cup\cdots\cup L_{n-1}$ of pre-specified edges. The individual sets that form $E_{Q}$ are defined as follows:

For $1\leq i\leq n-1$ and $1\leq t\leq l - 4i$ we define the \emph{zig-edges} $a_{i,t}=\{p_{t},q_{t+4i}\}$ and the \emph{zag-edges} $b_{i,t}=\{q_{t+4i},p_{t+1}\}$. For $i=0$, we define other zig- and zag-edges by $a_{0,t}=\{p_{t},q_{t+1}\}$ and $\ b_{0,t}=\{q_{t},p_{t}\}$, where this time $1\leq t\leq l-1$. For $i\in[n]$, we define the \emph{zig-zag} $L_{i}=\{a_{i,1},\ldots,a_{i,l-w},$ $b_{i,1},\ldots,b_{i,l-w}\}$ with $w=1$ for $i = 0$, and $w = 4i$ otherwise.

Next, we add the set of \emph{completion edges} $S$:
\begin{align*}
	S &=\{\{p_{1},q_{t}\}\mid 1\leq t\leq4(n-1)\}\cup\{\{p_{l-t},q_{l}\}\mid0\leq t\leq4(n-1)+2\}
\end{align*}
\indent Finally, we add the \emph{frame edges} $F=\{\{p_{i},p_{i+1}\}\mid i<l\}\cup\{\{q_{i},q_{i+1}\}\mid i<l \}$.
\end{definition}

\begin{figure}[!htb]
	\begin{center}
		\includegraphics[scale=0.9]{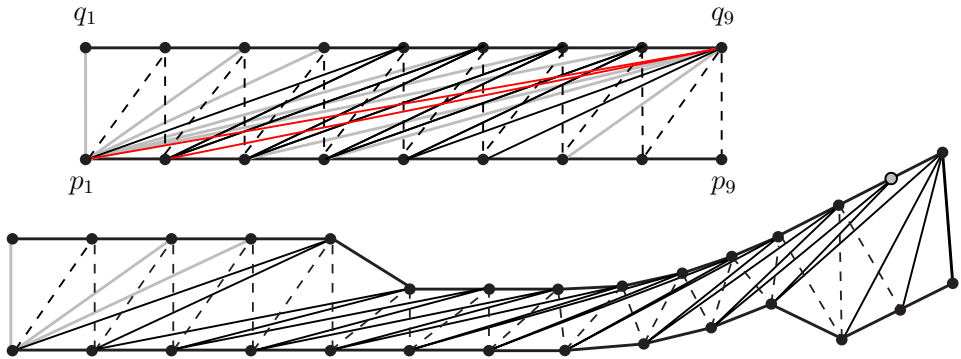}
	\end{center}
	\caption{(Top) A pipe with $3$ states and $l=9$. Thick black edges constitute $F$, thick gray edges constitute $S$, red edges are $L_{2}$, solid thin black edges are zig-zag $L_{1}$, dashed edges are zig-zag $L_{0}$. (Bottom) A stretched and bent wire with a terminal gadget attached to it.}
	\label{c-tri:sections:t-paths:figs:26}
\end{figure}

It is clear that \emph{any} triangulation $T$ of a pipe $Q$ contains exactly one zig-zag $L_{i}$, for $i\in[n]$, since different zig-zags lines cross. The sets $S, F$ help to complete a triangulation of $Q$ whenever zig-zag $L_{i}$ is present. If $L_{i}\subseteq T$, we say that $Q$ ``carries'' the value $i$ in $T$. Note that $F \subseteq T$ holds for every triangulation $T$ of $Q$. We cannot say the same about $S$ however.

A pipe with $n$ states will always be ``horizontal'', \emph{i.e.}, it will not turn in any other direction. This is required for the final set of points to feature a bounded number of onion layers.

In our construction we will also require vertical connections between pipes. These are obtained by \emph{wires}, which are pipes with two states. Since they feature only two states, wires can be stretched by arbitrary factors, and bent by arbitrary angles, while increasing their length only by a constant additive term. This is shown in Fig.~\ref{c-tri:sections:t-paths:figs:26}. For wires, we relabel the values $0$ and $1$ by $\textup{false}$ and $\textup{true}$ respectively.

The remaining gadgets for our reduction are specified and defined as follows: 

\begin{description}
\item [Or-gadget.] This gadget is connected to two input wires $W_{1},W_{2}$, and to an output wire $W_{3}$, as shown in Figure~\ref{c-tri:sections:t-paths:figs:27}. We have that: (\oldstylenums{1}) If one of $W_{1}$ or $W_{2}$ carries $\textup{true}$ in some restricted\footnote{Restricted w.r.t.~the shown adjacencies.} triangulation $T$ of the gadget, then $W_{3}$ may carry $\textup{true}$. (\oldstylenums{2}) If $W_{3}$ carries $\textup{true}$ in $T$, then at least one of $W_{1}$ or $W_{2}$ must necessarily carry $\textup{true}$.

\begin{figure}[!htb]
	\begin{center}
		\includegraphics[scale=0.9]{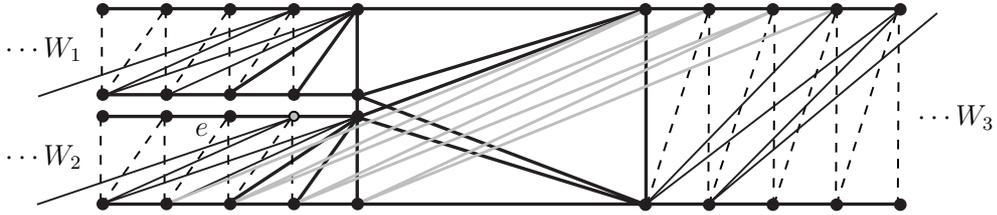}
	\end{center}
	\caption{The or-gadget. The gray edges from $W_{2}$ to $W_{3}$ are ``transfer edges''. An analogous set of edges is also present from $W_{1}$ to $W_{3}$, but suppressed in this figure to improve legibility.}
	\label{c-tri:sections:t-paths:figs:27}
\end{figure}

\item [A terminal-gadget.] This gadget can be attachted to a wire $W$, replacing its ``end part'' as exemplified in the bottom part of Figure~\ref{c-tri:sections:t-paths:figs:26}. It admits a triangulation iff $W$ carries $\textup{true}$.
		
\item [A tester-gadget.] This gadget is connected to a pipe $Q$, for value $i$ at position $t$, between $a_{i,t}$ and $b_{i,t}$, and has an output wire $W$, see to the left in Figure~\ref{c-tri:sections:t-paths:figs:28}. We have that: (\oldstylenums{1}) If $Q$ carries $i$ in some restricted triangulation $T$ of $Q$, then $W$ may carry $\textup{true}$. (\oldstylenums{2}) If $W$ carries $\textup{true}$, then $Q$ must carry $i$ in $T$.

\item [A crossing-gadget.] This is a more intricate gadget which allows an input wire $V$ to intersect a pipe $Q$, leaving it as an output wire $W$. The value carried by $Q$ is not influenced by $V$. We have that: (\oldstylenums{1}) If $V$ carries $\textup{true}$ in some restricted triangulation $T$ of the gadget, then $W$ may carry true. (\oldstylenums{2}) If $W$ carries $\textup{true}$, then $V$ must necessarily carry $\textup{true}$. 

As shown in the middle in Figure~\ref{c-tri:sections:t-paths:figs:28}, $V$ enters the crossing-gadget from the top. If $V$ intersects $Q$ between points $q_{t}$ and $q_{t+1}$ then a new point $r$ collinear with those two points is added to $Q$. Wire $V$ will now enter $Q$ between $r$ and $q_{t+1}$ instead, as shown in the middle in Figure~\ref{c-tri:sections:t-paths:figs:28}. Let us assume that $Q$ is an $n$-state pipe, and consider the set $S$ formed by the points $p_{u}$ such that $a_{i,u}$ is a zig-edge, of zig-zag $L_{i}$, adjacent to $q_{t+1}$, with $0\leq i\leq n-1$. By definition of $a_{i,u}$ we have that $u = t - 4i + 1$ for $1\leq i\leq n-1$, and $u = t$ for $i = 0$. There will be an output wire $W_{i}$, for zig-zag $L_{i}$, which will go out from $Q$ between $p_{u}\in S$ and $p_{u+1}$. 

Since pipes and wires are purely combinatorial objects, we have some freedom to move their points without affecting the adjacencies between the $p$'s and $q$'s, and without losing collinearities. Thus we will move all the $p$ points of $Q$ from $p_{t - 4(n-1) + 1}$ to $p_{t}$ to the right, and condense them in such a way that we keep their linear order, thus we also keep the planarity of the zig-zags $L_{i}$, $0\leq i\leq n-1$. The condensing part is also done in such a way that the following empty convex quadrilateral $C_{u}^{i}$ for zig-zag $L_{i}$ at $p_{u}\in S$ exists: Both diagonals of $C_{u}^{i}$ have negative slopes. One diagonal of $C_{u}^{i}$ is formed by $r$ and $p_{u+1}$. The other diagonal of $C_{u}^{i}$ is form by the point $\alpha_{u}$ of $V$, which is vertically aligned with $r$ and lies three points behind $r$ on $V$, and the point $\beta_{u}$ of $W_{i}$ which is vertically aligned with $p_{u+1}$ and lies three points ahead of $p_{u+1}$ on $W_{i}$. Points $\alpha_{u}$ and $\beta_{u}$, for $u = t - 3$, can be seen in the middle and to the right in Figure~\ref{c-tri:sections:t-paths:figs:28}. Observe that this re-arrangement of elements is always possible.

Now, the zig-edge $a_{i,u}$ of $L_{i}$ adjacent to $q_{t+1}$ is replaced by the edges $a^{\prime}_{i,u} = \{p_{u}, r\}$ and $\{p_{u+1}, r\}$. The latter edge is a diagonal of $C_{u}^{i}$ and is shown in red in Figure~\ref{c-tri:sections:t-paths:figs:28}. The rest of the adjacencies of $L_{i}$ remains the same.

Intuitively speaking, the red edge $\{p_{u+1}, r\}$ will help $V$ to transmit $\textup{false}$ to $W_{i}$, as seen to the right in Figure~\ref{c-tri:sections:t-paths:figs:28} for $i = 1$. Thus we also need to add the edges that will help $V$ to transmit $\textup{true}$ to $W_{i}$. Those edges are shown on solid black for $i = 1$ to the right in Figure~\ref{c-tri:sections:t-paths:figs:28}. The adjacencies are equivalent for any other $0\leq i\leq n-1$. Observe that all these adjacencies intersect neither $a^{\prime}_{i,u}$ nor $b_{i,u}$.

Finally, the output wires $W_0,\ldots,W_{n-1}$ are connected to a chain of or-gadgets, as shown in the middle in Figure~\ref{c-tri:sections:t-paths:figs:28}, whose output is precisely the output wire $W$.
\end{description}

\begin{figure}[!htb]
	\begin{center}
		\includegraphics[scale=0.9]{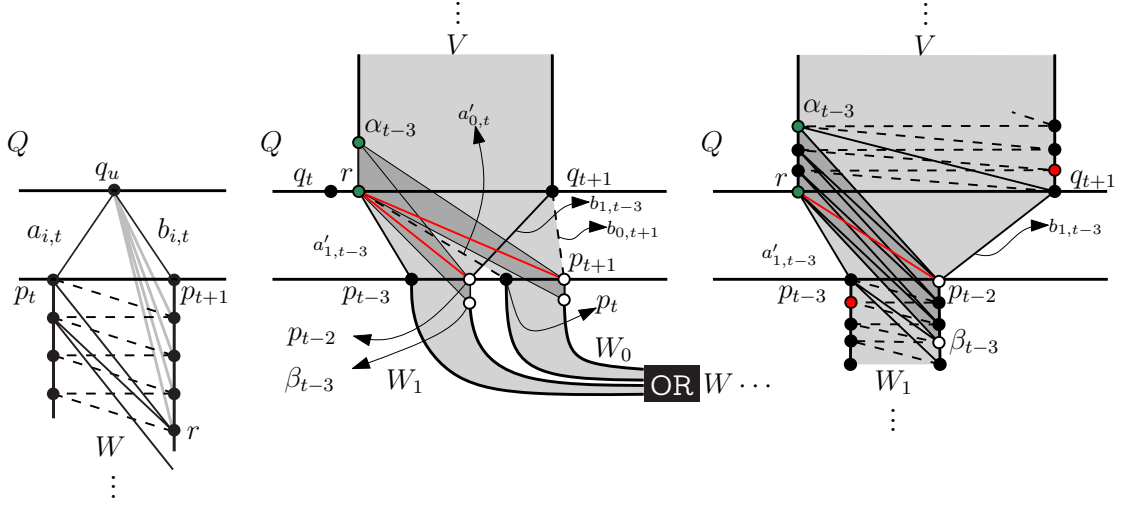}
	\end{center}
	\caption{To the left the tester-gadget for $i$ at $t$. $Q$ is modified by shifting, for $k>0$, all $p_{t+k}$ and $q_{u+k}$ to the right until the triangle $r,p_{t+1},q_{u}$ is oriented counter-clockwise. In the middle a crossing between pipe $Q$ and input wire $V$ which becomes output wire $W$. To the right the details of the crossing for $i = 1$ at $p_{t-3}$.}
	\label{c-tri:sections:t-paths:figs:28}
\end{figure}

%\newpage 

\subsection{Formal proofs}

\begin{lemma}\label{c-tri:sections:hardness:lemmas:1}
	All the previous gadgets fulfill their specifications.
\end{lemma}
\begin{proof}
\begin{description}
\item [{Or-gadget:}] (\oldstylenums{1}) Assume without loss of generality that $W_{2}$ carries $\textup{true}$ in some restricted triangulation $T$ of the gadget. Observe that $W_{3}$ carries $\textup{true}$ or $\textup{false}$ in $T$ depending on whether the transfer edges, shown in gray in Figure~\ref{c-tri:sections:t-paths:figs:27}, are chosen. For (\oldstylenums{2}) note that if $W_{3}$ carries $\textup{true}$ in $T$, the transfer edges from either $W_{1}$ or $W_{2}$, say $W_{2}$ without loss of generality, must be present. If $W_{2}$ carried $\textup{false}$, it can do it only up to edge $e$ shown in Figure~\ref{c-tri:sections:t-paths:figs:27}, since all following edges intersect transfer edges. But then, the gray point fails to be part of a triangle in \emph{any} restricted triangulation of the gadget, thus $W_{2}$ must necessarily carry $\textup{true}$ in $T$.

\item [{Terminal:}] If $W$ carries $\textup{true}$, the terminal is triangulated as shown at the bottom of Figure~\ref{c-tri:sections:t-paths:figs:26}. However, if $W$ carries $\textup{false}$, then the gray point shown in the same figure fails to be in a triangle of the restricted triangulation of the gadget.

\item [{Tester:}] For (\oldstylenums{1}) assume $Q$ carries $i$ in some restricted triangulation $T$. Since no gray edges in the tester intersect $L_{i}$, they can be added to $T$ or not. That would make $W$ carry $\textup{true}$ or $\textup{false}$ respectively. For (\oldstylenums{2}) given some restricted triangulation $T$, in which $W$ carries $\textup{true}$, all gray edges must be present in $T$. But the gadget is designed such that for every $0\leq j\neq i\leq n-1$, there is an edge $e\in L_{j}$ of $Q$ that intersects both $a_{i,t}$ and $b_{i,t}$, and thus all gray edges. Therefore $e\notin T$, and hence $L_{j}\nsubseteq T$, forcing $L_{i}\subseteq T$.

\item [{Crossing:}] (\oldstylenums{1})
Let $T$ be a restricted triangulation of the whole gadget in which $Q$ carries $i$. Observe that if either, $V$ or $W_{i}$, carries $\textup{true}$ in $T$, then the black solid edges that cross $Q$ from $V$ to $W_{i}$, shown to the right in Figure~\ref{c-tri:sections:t-paths:figs:28} for $i = 1$, must be present. Those edges in turn imply that the other gadget must necessarily carry $\textup{true}$ in $T$ as well, otherwise the red points shown to the right in Figure~\ref{c-tri:sections:t-paths:figs:28} will fail to be part of $T$, which would give us a contradiction since $T$ is a triangulation. Thus $V$ carries $\textup{true}$ iff $W_{i}$ carries $\textup{true}$, as long as $Q$ carries $i$. By using the chain of or-gadgets that the $W_{j}$'s are connected to we could leave the output wire $W$ carrying $\textup{true}$. (\oldstylenums{2}) Assume that $W$ carries $\textup{true}$ and $Q$ carries $i$ in $T$. Observe that in the chain of or-gadgets that the $W_{j}$'s are connected to, we can always force to transmit $\textup{true}$ from $W$ to $W_{i}$, while we transmit $\textup{false}$ to every other $W_{j}$, $0\leq j\neq i\leq n-1$. This in turn will force the edges that cross $Q$ from $W_{i}$ to $V$ to be included in $T$, the black solid edges that cross $Q$ from $W_{i}$ to $V$ shown to the right in Figure~\ref{c-tri:sections:t-paths:figs:28} for $i = 1$. This will make $V$ carry $\textup{true}$.

A triangulation is also possible if $V$ and $W$ carry $\textup{false}$. If $Q$ carries $i$ in $T$, then the edges $a^{\prime}_{i,u}=\{p_u,r\}$ and $\{p_{u+1},r\}$, for some $t - 4(n-1) + 1\leq u\leq t$, are also present in $T$. Thus we transmit $\textup{false}$ from $W$ to every $W_{j}$, $0\leq j\leq n-1$. However, as we said before, the red edge $\{p_{u+1},r\}$ will help to transmit $\textup{false}$ from $W_{i}$ to $V$ through $Q$. Thus $V$ would also carry $\textup{false}$ in $T$.\qedhere
\end{description}
\end{proof}

Theorem~\ref{c-tri:theorems:fpt} follows from the following lemma:

\begin{lemma}\label{c-tri:sections:hardness:lemmas:2}
\label{lem:composition}$G_{A}$ has $O(k)$ onion layers and admits a triangulation iff $A$ admits a hitting set of size $\leq k$.
\end{lemma}
\begin{proof}Consider the number of different $y$-coordinates of $\setp$. This is an upper bound for the number of onion layers of $\setp$. The pipes contribute $2k$ different $y$-coordinates. Every other gadget features $O(1)$ different $y$-coordinates. Each wire can be stretched and bent with $O(1)$ overhead, thus giving $O(1)$ different $y$-coordinates. Since the points in stripes $B_{i,1},\ldots, B_{i,m}$ are aligned by their $y$-coordinates, each set $B_{i,1}\cup\cdots\cup B_{i,m}$ has $O(1)$ different $y$-coordinates. This totals to $2k+O(k)=O(k)$ different $y$-coordinates among all points in $\setp$.

Given a hitting set $H=\{x_{1},\ldots,x_{k}\}$ of $k$ elements, we construct a triangulation that uses only edges from $E$ as follows: For every $i\leq k$, make $Q_{i}$ carry $x_{i}$. For every $j\leq m$ pick some $x = x_{i}\in H$ such that $x\in S_{j}$. In stripe $B_{i,j}$ triangulate the output wire of the tester for $x$ to carry $\textup{true}$, and transmit this $\textup{true}$ value along the or-gadgets of $B_{i,j}$. When crossing a pipe $Q_{z}$, with $z > i$, the $\textup{true}$ value will get transmitted through the output wire $W_{z}$ of the corresponding crossing-gadget. The $\textup{true}$ value will eventually reach the terminal of $\mathcal{B}_{j}$, which can then be triangulated without problems.

On the other hand, the values $H=\{x_{1},\ldots,x_{k}\}$ carried by the pipes $Q_{1},\ldots,Q_{k}$ in any restricted triangulation of $G_{A}$ form a hitting set. To see this, observe that every terminal must be triangulated, so the wire of every block $\mathcal{B}_{j}$ must carry $\textup{true}$ at some place. Thus, the output of some or-gadget in $\mathcal{B}_{j}$ must carry $\textup{true}$. Consider the first or-gadget that fulfills this top-down, and say it lies in stripe $B_{i,j}$. This or-gadget must be connected to a tester that outputs $\textup{true}$. This implies $x_{i}\in S_{j}$ and $H\cap S_{j}\neq\emptyset$.
\end{proof}

\section{Experimental results on counting triangulations}\label{c-tri:sections:experiments}

\newcommand{\tablesOne}[9]{
	{\color{black} #8} & {\color{black} #9} & \multicolumn{1}{l}{$#1$} & $#7$ & \multicolumn{1}{l}{{\color{black} #2}} &  \multicolumn{1}{l}{{\color{black} #3}} & \multicolumn{1}{l}{{\color{black} #5}} & \multicolumn{1}{l}{{\color{black} #6}}\\
}
\newcommand{\tablesTwo}[9]{
	{\color{black} #1} & {\color{black} #2} & \multicolumn{1}{l}{$#3$} & $#4$ & \multicolumn{1}{l}{$#5$} & $#6$ & $#9$\\
}
\newcommand{\gridOne}[9]{
	{\color{black} #1} & {\color{black} #2} & {\color{black} #9} & \multicolumn{1}{l}{$#3$} & $#4$ & \multicolumn{1}{l}{{\color{black} #5}} & \multicolumn{1}{l}{{\color{black} #7}}\\
}
\newcommand{\gridTwo}[8]{
	{\color{black} #1} & {\color{black} #2} & {\color{black} #7} & \multicolumn{1}{l}{$#3$} & $#4$ & $#8$\\
}

We have implemented the sn-path algorithm presented in~\S~\ref{c-tri:sections:triple-paths}, and in this section we compare it with the algorithm presented in~\cite{ray-seidel}. All experiments were run on a server generously provided by Prof. Bernd Finkbeiner, head of the Reactive Systems group at Saarland University. All implementations are single-threaded, so all algorithms were run a on single core of a dual-core processor AMD Opteron at 2.6 Ghz. Linux was the used operating system, and the amount of RAM available was 122 GB. Finally, all implementations use the GMP library to handle big numbers. All statistics reported here were obtained from the output of the command `time -v'.

The main idea behind the experiments was to obtain evidence of the practical limits of the algorithms, thus they are really provided without statistical analysis. If we denote the number of points by $n$, the number of onion layers by $k$, and the size of the convex hull by $h$, we were interested in knowing for different values of those parameters what are the largest sets of points we can solve. Also, besides providing the number of triangulations, we also provide: (\oldstylenums{1}) Total number of sub-problems generated by each algorithm, (\oldstylenums{2}) Memory consumption, and (\oldstylenums{3}) Total running time. Since we used memoization in all algorithms, the total number of sub-problems is just the size of the database at the end of execution.

We have four kinds of sets of points, of selected cardinalities, we ran the algorithms on: (\oldstylenums{1}) Sets of points having three onion layers. (\oldstylenums{2}) Sets of points generated in a square (\oldstylenums{3}) Sets of points having the largest possible number of onion layers, w.r.t.~the cardinality of the set (\oldstylenums{4}) Grids. Sets (\oldstylenums{1}), (\oldstylenums{2}) and (\oldstylenums{3}) were generated at random. For (\oldstylenums{1}) we generated random points on three concentric circles and we only kept configurations having three onion layers.

Table~\ref{c-tri:sections:experiments:table:review} summarizes the largest sets of points, of each type, that we were able to solve within 140 hours, the complete results can be seen in the tables at the end of the paper. Results for (\oldstylenums{1}) are shown in Tables~\ref{c-tri:sections:experiments:table:3cl:1} and~\ref{c-tri:sections:experiments:table:3cl:2}, for (\oldstylenums{2}) in Tables~\ref{c-tri:sections:experiments:table:square:1} and~\ref{c-tri:sections:experiments:table:square:2}, for (\oldstylenums{3}) in Tables~\ref{c-tri:sections:experiments:table:nestedTri:1} and~\ref{c-tri:sections:experiments:table:nestedTri:2}, and for (\oldstylenums{4}) in Tables~\ref{c-tri:sections:experiments:table:grid:1} and~\ref{c-tri:sections:experiments:table:grid:2}. In the tables the algorithm of~\cite{ray-seidel} is called ``ray-seidel'', and our algorithm simply ``sn-paths''. We did not run ray-seidel on sets of kind (\oldstylenums{4}) due to degeneracy; this functionality was not implemented. All columns are self-explanatory except for the columns ``Base'' and ``Exp''. The former refers to the base $c$, truncated to two decimal digits, of a number expressed as $c^{n}$. The latter refers to the term $d$, also truncated to two decimal digits, of a number expressed as $n^{d\cdot k}$, this makes sense for sn-paths since we know that for fixed $k$ the running time is $n^{O(k)}$.

{\arrayrulecolor{black}
	\begin{table}[!htb]
		\centering
		\begin{tabular}{lcccc}
			& \multicolumn{4}{c}{{\color{black} \# Points}}\\
			\\[-2.5ex]\cline{2-5}\\[-2ex]
			& {\color{black} (\oldstylenums{1})} & {\color{black} (\oldstylenums{2})} & {\color{black} (\oldstylenums{3})} & {\color{black} (\oldstylenums{4})}\\
			\\[-2.5ex]\hline\\[-2ex]
			{\color{black} ray-seidel} & {\color{black} 43} & {\color{black} 43} & {\color{black} 34} & {\color{black} NA}\\
			\\[-2.5ex]\hline\\[-2ex]
			{\color{black} sn-paths} & {\color{black} 80} & {\color{black} 43} & {\color{black} 28} & {\color{black} 6x17}\\
			\\[-2.5ex]\hline
		\end{tabular}
		\caption{Largest sets of points of kind $(i)$, $1\leq i\leq 4$, solved by each algorithm within 140 hours.}
		\label{c-tri:sections:experiments:table:review}
	\end{table}
}

Since we are interested in the largest $n$ we can solve, we started the experiments with at least 25 points, below this threshold all algorithms perform very well, where ray-seidel is notably the fastest, giving answers in at most a couple of seconds, and sn-paths the slowest for $k = 8$. All empty entries, except for the last entry of sn-paths in Tables~\ref{c-tri:sections:experiments:table:square:1} and~\ref{c-tri:sections:experiments:table:square:2}, mean that the corresponding algorithm consumed all available RAM memory \emph{before} finishing the corresponding set of points. In the same sense, one complete empty row means that \emph{no} algorithm managed to finish the corresponding set of points. The last entry of sn-paths in tables~\ref{c-tri:sections:experiments:table:square:1} and~\ref{c-tri:sections:experiments:table:square:2} was explicitly stopped due to its potentially large running time.

To verify the correctness of the algorithms we ran them on configurations available in~\cite{oswin-web,volker-ziegler}, for which an answer is known via other algorithms. We also run them on sets of points in convex position, there the number of triangulations is a Catalan number. In all cases the two algorithms confirmed the known answers.

{\arrayrulecolor{black}
	\begin{sidewaystable}[!htb]
		\centering\small
			\begin{tabular}{cccccccc}
				& & & & \multicolumn{2}{c}{{\color{black} Time in hh:mm:ss.ms}} & \multicolumn{2}{c}{{\color{black}RAM in Mb}}\\
				\\[-2.5ex]\cline{5-8}\\[-2ex]
					$n$ & $h$ & {\color{black}\#Triangulations} & {\color{black} Base} & {\color{black}ray-seidel} & {\color{black} sn-paths} & {\color{black} ray-seidel} & {\color{black} sn-paths}\\
					\\[-2.5ex]\hline\\[-2ex]
					{\color{black} 30} & {\color{black} 10} & \multicolumn{1}{l}{$161014656152655441$} & $\approx 3.74$ & \multicolumn{1}{l}{{\color{black} 20.18}} & \multicolumn{1}{l}{{\color{black} 55.77}} & \multicolumn{1}{l}{{\color{black} 303}} &  \multicolumn{1}{l}{{\color{black} 33}}\\
					 & {\color{black} 10} & \multicolumn{1}{l}{$312513373686594183$} & $\approx 3.82$ & \multicolumn{1}{l}{{\color{black} 1:07.72}} & \multicolumn{1}{l}{{\color{black} 1:09.17}} & \multicolumn{1}{l}{{\color{black} 1160}} &  \multicolumn{1}{l}{{\color{black} 38}}\\
					\\[-2.5ex]\hline\\[-2ex]
					\tablesOne{32155601714553665796}{3:50.96}{2:19.19}{}{2762}{62}{\approx 3.90}{33}{11}
					\tablesOne{68598010833407738067}{58.43}{2:30.94}{}{857}{62}{\approx 3.99}{}{11}
					\\[-2.5ex]\hline\\[-2ex]
					\tablesOne{9334947679230323509429}{1:53.60}{3:43.51}{}{1521}{90}{\approx 3.92}{37}{12}
					\tablesOne{31113068813012076443512}{3:20.41}{4:42.24}{}{2465}{97}{\approx 4.05}{}{12}
					\\[-2.5ex]\hline\\[-2ex]
					\tablesOne{2642143054680217856074126}{18:12.16}{8:10.26}{}{12557}{131}{\approx 4.07}{40}{14}
					\tablesOne{2903778262295075928823011}{7:01.40}{9:50.98}{}{4325}{149}{\approx 4.08}{}{15}
					\\[-2.5ex]\hline\\[-2ex]
					\tablesOne{452371697808162396583055656}{1:34:48}{21:39.66}{}{53263}{243}{\approx 4.16}{43}{14}
					\tablesOne{461550214764369881018564051}{}{19:25.53}{}{}{242}{\approx 4.16}{}{14}
					\\[-2.5ex]\hline\\[-2ex]
					\tablesOne{157759710540671985436621922639}{}{32:46.68}{}{}{363}{\approx 4.18}{47}{16}
					\tablesOne{341037585238678346710372748758}{}{39:33.42}{}{}{420}{\approx 4.24}{}{15}
					\\[-2.5ex]\hline\\[-2ex]
					\tablesOne{54782168649020627430413001433261}{}{1:06:53}{}{}{606}{\approx 4.31}{50}{16}
					\tablesOne{158997592723683977758501079915910}{}{1:07:19}{}{}{553}{\approx 4.40}{}{16}
					\\[-2.5ex]\hline\\[-2ex]
					\tablesOne{383051932722566765683591748023039}{}{7:51:53}{}{}{2006}{\approx 4.56}{60}{21}
					\tablesOne{0004428}{}{}{}{}{}{}{}{}
					\tablesOne{190030780266337926700033771493586}{}{8:59:36}{}{}{2063}{\approx 4.69}{}{20}
					\tablesOne{96361338}{}{}{}{}{}{}{}{}
					\\[-2.5ex]\hline\\[-2ex]
					\tablesOne{481423578642758908977651277967532}{}{22:59:43}{}{}{3937}{\approx 4.64}{70}{25}
					\tablesOne{53695164862078}{}{}{}{}{}{}{}{}
					\tablesOne{224411547729672469823709078962020}{}{26:54:49}{}{}{4506}{\approx 4.74}{}{23}
					\tablesOne{667530864087588}{}{}{}{}{}{}{}{}
					\\[-2.5ex]\hline\\[-2ex]
					\tablesOne{396978851668966053957582788796899}{}{81:29:23}{}{}{8932}{\approx 4.81}{80}{26}
					\tablesOne{3403864755422464030090}{}{}{}{}{}{}{}{}
					\tablesOne{185279982277182715207126583259662}{}{90:34:22}{}{}{9617}{\approx 4.90}{}{26}
					\tablesOne{53393485474040452858832}{}{}{}{}{}{}{}{}
				\\[-2.5ex]\hline
			\end{tabular}
			\caption{$n$ random points on a square, having three onion layers and $h$ points on their convex hull.}
			\label{c-tri:sections:experiments:table:3cl:1}
	\end{sidewaystable}
	
%	\begin{sidewaystable}[p]
	\begin{table}
		\centering\small
			\begin{tabular}{ccccccc}
				& & \multicolumn{5}{c}{{\color{black} \# Sub-problems}}\\
				\\[-2.5ex]\cline{3-7}\\[-2ex]
				\multicolumn{1}{c}{$n$} & \multicolumn{1}{c}{$h$} & \multicolumn{1}{c}{\color{black}ray-seidel} & {\color{black} Base} & \multicolumn{1}{c}{\color{black} sn-paths} & {\color{black} Base} & {\color{black} Exp}\\
					\\[-2.5ex]\hline\\[-2ex]
					\tablesTwo{30}{10}{2050514}{\approx 1.62}{215732}{\approx 1.50}{}{}{\approx 1.20}
					\tablesTwo{}{10}{7879754}{\approx 1.69}{246657}{\approx 1.51}{}{}{\approx 1.21}
					\\[-2.5ex]\hline\\[-2ex]
					\tablesTwo{33}{11}{18992928}{\approx 1.66}{405580}{\approx 1.47}{}{}{\approx 1.23}
					\tablesTwo{}{11}{5812991}{\approx 1.60}{410357}{\approx 1.47}{}{}{\approx 1.23}
					\\[-2.5ex]\hline\\[-2ex]
					\tablesTwo{37}{12}{10027300}{\approx 1.54}{575255}{\approx 1.43}{}{}{\approx 1.22}
					\tablesTwo{}{12}{16250100}{\approx 1.56}{626274}{\approx 1.43}{}{}{\approx 1.23}
					\\[-2.5ex]\hline\\[-2ex]
					\tablesTwo{40}{14}{82635240}{\approx 1.57}{866278}{\approx 1.40}{}{}{\approx 1.23}
					\tablesTwo{}{15}{28333612}{\approx 1.53}{982791}{\approx 1.41}{}{}{\approx 1.24}
					\\[-2.5ex]\hline\\[-2ex]
					\tablesTwo{43}{14}{347603518}{\approx 1.57}{1604269}{\approx 1.39}{}{}{\approx 1.26}
					\tablesTwo{}{14}{}{}{1591423}{\approx 1.39}{}{}{\approx 1.26}
					\\[-2.5ex]\hline\\[-2ex]
					\tablesTwo{47}{16}{}{}{2287764}{\approx 1.36}{}{}{\approx 1.26}
					\tablesTwo{}{15}{}{}{2720786}{\approx 1.37}{}{}{\approx 1.28}
					\\[-2.5ex]\hline\\[-2ex]
					\tablesTwo{50}{16}{}{}{3631525}{\approx 1.35}{}{}{\approx 1.28}
					\tablesTwo{}{16}{}{}{3998798}{\approx 1.35}{}{}{\approx 1.29}
					\\[-2.5ex]\hline\\[-2ex]
					\tablesTwo{60}{21}{}{}{12527119}{\approx 1.31}{}{}{\approx 1.33}
					\tablesTwo{}{20}{}{}{13076694}{\approx 1.31}{}{}{\approx 1.33}
					\\[-2.5ex]\hline\\[-2ex]
					\tablesTwo{70}{25}{}{}{23762305}{\approx 1.27}{}{}{\approx 1.33}
					\tablesTwo{}{23}{}{}{27937551}{\approx 1.27}{}{}{\approx 1.34}
					\\[-2.5ex]\hline\\[-2ex]
					\tablesTwo{80}{26}{}{}{54047260}{\approx 1.24}{}{}{\approx 1.35}
					\tablesTwo{}{26}{}{}{58561612}{\approx 1.25}{}{}{\approx 1.36}
				\\[-2.5ex]\hline
			\end{tabular}
			\caption{Number of sub-problems generated by the configurations (entry-wise) presented in Table~\ref{c-tri:sections:experiments:table:3cl:1}.}
			\label{c-tri:sections:experiments:table:3cl:2}
	\end{table}
%	\end{sidewaystable}
}

{\arrayrulecolor{black}
\begin{sidewaystable}[p]
	\centering\small
		\begin{tabular}{ccccccccc}
			& & & & & \multicolumn{2}{c}{{\color{black} Time in hhh:mm:ss.ms}} & \multicolumn{2}{c}{{\color{black}RAM in Mb}}\\
			\\[-2.5ex]\cline{6-9}\\[-2ex]
			$n$ & $k$ & $h$ & {\color{black}\#Triangulations} & {\color{black} Base} & {\color{black}ray-seidel} & {\color{black} sn-paths} & {\color{black} ray-seidel} & {\color{black} sn-paths}\\
			\\[-2.5ex]\hline\\[-2ex]
			{\color{black} 30} & {\color{black} 5} & {\color{black} 9} & \multicolumn{1}{l}{$29762284427845618$} & $\approx 3.54$ & \multicolumn{1}{l}{{\color{black} 7.60}} &  \multicolumn{1}{l}{{\color{black} 3:56.61}} & \multicolumn{1}{l}{{\color{black} 130}} &  \multicolumn{1}{l}{{\color{black} 141}}\\
			 & {\color{black} 6} & {\color{black} 7} & \multicolumn{1}{l}{$54648952555202115$} & $\approx 3.61$ & \multicolumn{1}{l}{{\color{black} 30.69}} & \multicolumn{1}{l}{{\color{black} 16:17.12}} & \multicolumn{1}{l}{{\color{black} 470}} & \multicolumn{1}{l}{{\color{black} 535}}\\
			\\[-2.5ex]\hline\\[-2ex]
			\multicolumn{1}{c}{{\color{black} 33}} & \multicolumn{1}{c}{{\color{black} 5}} & \multicolumn{1}{c}{{\color{black} 11}} & \multicolumn{1}{l}{$8830953374442248378$} & \multicolumn{1}{c}{{\color{black} $\approx 3.75$}} &  \multicolumn{1}{l}{{\color{black} 34.80}} & \multicolumn{1}{l}{{\color{black} 14:47.26}} & \multicolumn{1}{l}{{\color{black} 643}} & \multicolumn{1}{l}{{\color{black} 394}}\\
			\multicolumn{1}{c}{} & \multicolumn{1}{c}{{\color{black} 6}} & \multicolumn{1}{c}{{\color{black} 7}} & \multicolumn{1}{l}{$23407918365649149382$} & \multicolumn{1}{c}{{\color{black} $\approx3.86$}} & \multicolumn{1}{l}{{\color{black} 15.10}} &  \multicolumn{1}{l}{{\color{black} 1:10:34}} & \multicolumn{1}{l}{{\color{black} 288}} &  \multicolumn{1}{l}{{\color{black} 1292}}\\
			\\[-2.5ex]\hline\\[-2ex]
			\multicolumn{1}{c}{{\color{black} 37}} & \multicolumn{1}{c}{{\color{black} 5}} & \multicolumn{1}{c}{{\color{black} 11}} & \multicolumn{1}{l}{$8317197892568798832050$} & \multicolumn{1}{c}{{\color{black} $\approx 3.91$}} & \multicolumn{1}{l}{{\color{black} 3:35.27}} & \multicolumn{1}{l}{{\color{black} 1:15:00}} & \multicolumn{1}{l}{{\color{black} 2796}} &  \multicolumn{1}{l}{{\color{black} 1524}}\\
			\multicolumn{1}{c}{} & \multicolumn{1}{c}{{\color{black} 5}} & \multicolumn{1}{c}{{\color{black} 13}} & \multicolumn{1}{l}{$15347609782987966767248$} & \multicolumn{1}{c}{{\color{black} $\approx 3.97$}} & \multicolumn{1}{l}{{\color{black} 15:23.53}} &  \multicolumn{1}{l}{{\color{black} 2:16:32}} & \multicolumn{1}{l}{{\color{black} 10707}} &  \multicolumn{1}{l}{{\color{black} 1957}}\\
			\\[-2.5ex]\hline\\[-2ex]
			\multicolumn{1}{c}{{\color{black} 40}} & \multicolumn{1}{c}{{\color{black} 6}} & \multicolumn{1}{c}{{\color{black} 12}} & \multicolumn{1}{l}{$1146138971033715203926926$} & \multicolumn{1}{c}{{\color{black} $\approx 3.99$}} & \multicolumn{1}{l}{{\color{black} 25:42.43}} & \multicolumn{1}{l}{{\color{black} 13:29:43}} & \multicolumn{1}{l}{{\color{black} 18525}} &  \multicolumn{1}{l}{{\color{black} 8889}}\\
			\multicolumn{1}{c}{} & \multicolumn{1}{c}{{\color{black} 7}} & \multicolumn{1}{c}{{\color{black} 10}} & \multicolumn{1}{l}{$5050493282169462429012536$} & \multicolumn{1}{c}{{\color{black} $\approx 4.14$}} & \multicolumn{1}{l}{{\color{black} 1:35:45}} &  \multicolumn{1}{l}{{\color{black} 46:49:41}} & \multicolumn{1}{l}{{\color{black} 54128}} &  \multicolumn{1}{l}{{\color{black} 25533}}\\
			\\[-2.5ex]\hline\\[-2ex]
			\multicolumn{1}{c}{{\color{black} 43}} & \multicolumn{1}{c}{{\color{black} 6}} & \multicolumn{1}{c}{{\color{black} 10}} & \multicolumn{1}{l}{$981403313298259834292202925$} & \multicolumn{1}{c}{{\color{black} $\approx 4.24$}} & \multicolumn{1}{l}{{\color{black} 3:20:54}} & \multicolumn{1}{l}{{\color{black} 107:48:48}} & \multicolumn{1}{l}{{\color{black} 116506}} &  \multicolumn{1}{l}{{\color{black} 37407}}\\
			\multicolumn{1}{c}{} & \multicolumn{1}{c}{{\color{black} 7}} & \multicolumn{1}{c}{{\color{black} 8}} & \multicolumn{1}{c}{{\color{black}}} & \multicolumn{1}{c}{} & \multicolumn{1}{c}{{\color{black}}} &  \multicolumn{1}{c}{{\color{black}}} & \multicolumn{1}{c}{{\color{black}}} &  \multicolumn{1}{c}{{\color{black}}}\\
		\\[-2.5ex]\hline
		\end{tabular}
		\caption{$n$ random points on a square, having $k$ onion layers and $h$ points of their convex hull.}
		\label{c-tri:sections:experiments:table:square:1}
		
		\vspace{1cm}
		
		\begin{tabular}{cccccccc}
			& & & \multicolumn{5}{c}{{\color{black} \# Sub-problems}}\\
			\\[-2.5ex]\cline{4-8}\\[-2ex]
			$n$ & $k$ & $h$ & {\color{black} ray-seidel} & {\color{black} Base} & {\color{black} sn-paths} & {\color{black} Base} & {\color{black}  Exp}\\
			\\[-2.5ex]\hline\\[-2ex]
			\multicolumn{1}{c}{{\color{black} 30}} & \multicolumn{1}{c}{{\color{black} 5}} & \multicolumn{1}{c}{{\color{black} 9}} & \multicolumn{1}{l}{$854579$} & \multicolumn{1}{c}{{\color{black} $\approx 1.56$}} & \multicolumn{1}{l}{$947262$} & \multicolumn{1}{c}{{\color{black} $\approx 1.58$}} & $\approx 0.80$\\
			\multicolumn{1}{c}{} &  \multicolumn{1}{c}{{\color{black} 6}} & \multicolumn{1}{c}{{\color{black} 7}} & \multicolumn{1}{l}{$3150228$} & \multicolumn{1}{c}{{\color{black} $\approx 1.64$}} & \multicolumn{1}{l}{$3590878$} & \multicolumn{1}{c}{{\color{black} $\approx 1.65$}} & $\approx 0.73$\\
			\\[-2.5ex]\hline\\[-2ex]
			\multicolumn{1}{c}{{\color{black} 33}} & \multicolumn{1}{c}{{\color{black} 5}} & \multicolumn{1}{c}{{\color{black} 11}} & \multicolumn{1}{l}{$4245399$} & \multicolumn{1}{c}{{\color{black} $\approx 1.58$}} & \multicolumn{1}{l}{$2554063$} & \multicolumn{1}{c}{{\color{black} $\approx 1.56$}} & $\approx 0.84$\\
			\multicolumn{1}{c}{} & \multicolumn{1}{c}{{\color{black} 6}} & \multicolumn{1}{c}{{\color{black} 7}} & \multicolumn{1}{l}{$1907449$} & \multicolumn{1}{c}{{\color{black} $\approx 1.54$}} & \multicolumn{1}{l}{$8731943$} & \multicolumn{1}{c}{{\color{black} $\approx 1.62$}} & $\approx 0.76$\\
			\\[-2.5ex]\hline\\[-2ex]
			\multicolumn{1}{c}{{\color{black} 37}} & \multicolumn{1}{c}{{\color{black} 5}} & \multicolumn{1}{c}{{\color{black} 11}} & \multicolumn{1}{l}{$18477670$} & \multicolumn{1}{c}{{\color{black} $\approx 1.57$}} & \multicolumn{1}{l}{$9735430$} & \multicolumn{1}{c}{{\color{black} $\approx 1.54$}} & $\approx 0.89$\\
			\multicolumn{1}{c}{} & \multicolumn{1}{c}{{\color{black} 5}} & \multicolumn{1}{c}{{\color{black} 13}} & \multicolumn{1}{l}{$70483691$} & \multicolumn{1}{c}{{\color{black} $\approx 1.62$}} & \multicolumn{1}{l}{$12535632$} & \multicolumn{1}{c}{{\color{black} $\approx 1.55$}} & $\approx 0.90$\\
			\\[-2.5ex]\hline\\[-2ex]
			\multicolumn{1}{c}{{\color{black} 40}} & \multicolumn{1}{c}{{\color{black} 6}} & \multicolumn{1}{c}{{\color{black} 12}} & \multicolumn{1}{l}{$121049523$} & \multicolumn{1}{c}{{\color{black} $\approx 1.59$}} & \multicolumn{1}{l}{$56587195$} & \multicolumn{1}{c}{{\color{black} $\approx 1.56$}} & $\approx 0.80$\\
			\multicolumn{1}{c}{} & \multicolumn{1}{c}{{\color{black} 7}} & \multicolumn{1}{c}{{\color{black} 10}} & \multicolumn{1}{l}{$354717051$} & \multicolumn{1}{c}{{\color{black} $\approx 1.63$}} & \multicolumn{1}{l}{$155716531$} & \multicolumn{1}{c}{{\color{black} $\approx 1.60$}} & $\approx 0.73$\\
			\\[-2.5ex]\hline\\[-2ex]
			\multicolumn{1}{c}{{\color{black} 43}} & \multicolumn{1}{c}{{\color{black} 6}} & \multicolumn{1}{c}{{\color{black} 10}} & \multicolumn{1}{l}{$752596823$} & \multicolumn{1}{c}{{\color{black} $\approx 1.60$}} & \multicolumn{1}{l}{$239084256$} & \multicolumn{1}{c}{{\color{black} $\approx 1.56$}} & $\approx 0.85$\\
			\multicolumn{1}{c}{} & \multicolumn{1}{c}{{\color{black} 7}} & \multicolumn{1}{c}{{\color{black} 8}} & \multicolumn{1}{c}{\color{black}} & \multicolumn{1}{c}{\color{black}} & \multicolumn{1}{c}{\color{black}} & \multicolumn{1}{c}{\color{black}} & \\
			\\[-2.5ex]\hline
		\end{tabular}
		\caption{Number of sub-problems generated by the configurations (entry-wise) presented in Table~\ref{c-tri:sections:experiments:table:square:1}.}
		\label{c-tri:sections:experiments:table:square:2}
\end{sidewaystable}
}

{\arrayrulecolor{black}
	\begin{sidewaystable}
		\centering\small
			\begin{tabular}{cccccccc}
				& & & & \multicolumn{2}{c}{{\color{black} Time in hh:mm:ss.ms}} & \multicolumn{2}{c}{{\color{black}RAM in Mb}}\\ 
				\\[-2.5ex]\cline{5-8}\\[-2ex]
				\multicolumn{1}{c}{$k$} & $n$ & {\color{black}\#Triangulations} & \multicolumn{1}{c}{{\color{black} Base}} & {\color{black}ray-seidel} & {\color{black} sn-paths} & {\color{black} ray-seidel} & {\color{black} sn-paths}\\
			\\[-2.5ex]\hline\\[-2ex]
				\multicolumn{1}{c}{{\color{black} 9}} & \multicolumn{1}{c}{{\color{black} 25}} & \multicolumn{1}{l}{$248441701550196$} & \multicolumn{1}{c}{{\color{black} $\approx 3.76$}} & \multicolumn{1}{l}{{\color{black} 7.82}} &  \multicolumn{1}{l}{{\color{black} 1:36:07}} & \multicolumn{1}{l}{{\color{black} 154}} &  \multicolumn{1}{l}{{\color{black} 3828}}\\
			\multicolumn{1}{c}{} &  \multicolumn{1}{c}{{\color{black} 27}} & \multicolumn{1}{l}{$6632755933105064$} & \multicolumn{1}{c}{{\color{black} $\approx 3.85$}} & \multicolumn{1}{l}{{\color{black} 1:16.84}} & \multicolumn{1}{l}{{\color{black} 9:56:09}} & \multicolumn{1}{l}{{\color{black} 1239}} &  \multicolumn{1}{l}{{\color{black} 17406}}\\
			\\[-2.5ex]\hline\\[-2ex]
				\multicolumn{1}{c}{{\color{black} 10}} &  \multicolumn{1}{c}{{\color{black} 28}} & \multicolumn{1}{l}{$134806114688321888$} & \multicolumn{1}{c}{{\color{black} $\approx 4.09$}} & \multicolumn{1}{l}{{\color{black} 1:04.33}} & \multicolumn{1}{l}{{\color{black} 39:29:17}} & \multicolumn{1}{l}{{\color{black} 1130}} &  \multicolumn{1}{l}{{\color{black} 56197}}\\
				\multicolumn{1}{c}{} & \multicolumn{1}{c}{{\color{black} 28}} & \multicolumn{1}{l}{$259051751512786147$} & \multicolumn{1}{c}{{\color{black} $\approx 4.18$}} & \multicolumn{1}{l}{{\color{black} 58.55}} &  \multicolumn{1}{l}{{\color{black} 32:31:58}} & \multicolumn{1}{l}{{\color{black} 903}} &  \multicolumn{1}{l}{{\color{black} 41745}}\\
			\\[-2.5ex]\hline\\[-2ex]
				\multicolumn{1}{c}{{\color{black} 11}} &  \multicolumn{1}{c}{{\color{black} 32}} & \multicolumn{1}{l}{$188748482026800154083$} & \multicolumn{1}{c}{{\color{black} $\approx 4.30$}} & \multicolumn{1}{l}{{\color{black} 1:43:22}} & \multicolumn{1}{c}{{\color{black}}} & \multicolumn{1}{l}{{\color{black} 70647}} &  \multicolumn{1}{c}{{\color{black}}}\\
				\multicolumn{1}{c}{} & \multicolumn{1}{c}{{\color{black} 33}} & \multicolumn{1}{l}{$2680138023948109608080$} & \multicolumn{1}{c}{{\color{black} $\approx 4.46$}} & \multicolumn{1}{l}{{\color{black} 1:38:04}} &  \multicolumn{1}{c}{{\color{black}}} & \multicolumn{1}{l}{{\color{black} 62477}} &  \multicolumn{1}{c}{{\color{black}}}\\
			\\[-2.5ex]\hline\\[-2ex]
				\multicolumn{1}{c}{{\color{black} 12}} &  \multicolumn{1}{c}{{\color{black} 34}} & \multicolumn{1}{l}{$16605186163166445755560$} & \multicolumn{1}{c}{{\color{black} $\approx 4.50$}} & \multicolumn{1}{l}{{\color{black} 2:52:16}} & \multicolumn{1}{c}{{\color{black}}} & \multicolumn{1}{l}{{\color{black} 114676}} &  \multicolumn{1}{c}{{\color{black}}}\\
				\multicolumn{1}{c}{} & \multicolumn{1}{c}{{\color{black} 34}} & \multicolumn{1}{c}{{\color{black}}} & \multicolumn{1}{c}{{\color{black}}} & \multicolumn{1}{c}{{\color{black}}} &  \multicolumn{1}{c}{{\color{black}}} & \multicolumn{1}{c}{{\color{black}}} &  \multicolumn{1}{c}{{\color{black}}}\\
			\\[-2.5ex]\hline
			\end{tabular}
			\caption{$n$ random points having $k = \left\lceil\frac{n}{3}\right\rceil$ onion layers.}
			\label{c-tri:sections:experiments:table:nestedTri:1}
			
			\vspace{1cm}
			
		\begin{tabular}{cccccc}
			& & \multicolumn{4}{c}{{\color{black} \# Sub-problems}}\\
			\\[-2.5ex]\cline{3-6}\\[-2ex]
			\multicolumn{1}{c}{$k$} & $n$ & \multicolumn{1}{c}{\color{black}ray-seidel} & {\color{black} Base} & \multicolumn{1}{c}{\color{black} sn-paths} & {\color{black} Base}\\
			\\[-2.5ex]\hline\\[-2ex]
			\multicolumn{1}{c}{{\color{black} 9}} & \multicolumn{1}{c}{{\color{black} 25}} & \multicolumn{1}{l}{$1078399$} & \multicolumn{1}{c}{{\color{black} $\approx 1.74$}} & \multicolumn{1}{l}{$24811886$} & \multicolumn{1}{c}{{\color{black} $\approx 1.97$}}\\
			\multicolumn{1}{c}{} &  \multicolumn{1}{c}{{\color{black} 27}} & \multicolumn{1}{l}{$8738535$} & \multicolumn{1}{c}{{\color{black} $\approx 1.80$}} & \multicolumn{1}{l}{$110817524$} & \multicolumn{1}{c}{{\color{black} $\approx 1.98$}}\\
			\\[-2.5ex]\hline\\[-2ex]
			\multicolumn{1}{c}{{\color{black} 10}} & \multicolumn{1}{c}{{\color{black} 28}} & \multicolumn{1}{l}{$8015023$} & \multicolumn{1}{c}{{\color{black} $\approx 1.76$}} & \multicolumn{1}{l}{$347448787$} & \multicolumn{1}{c}{{\color{black} $\approx 2.01$}}\\
			\multicolumn{1}{c}{} & \multicolumn{1}{c}{{\color{black} 28}} & \multicolumn{1}{l}{$6303203$} & \multicolumn{1}{c}{{\color{black} $\approx 1.74$}} & \multicolumn{1}{l}{$266661064$} & \multicolumn{1}{c}{{\color{black} $\approx 1.99$}}\\
			\\[-2.5ex]\hline\\[-2ex]
			\multicolumn{1}{c}{{\color{black} 11}} & \multicolumn{1}{c}{{\color{black} 32}} & \multicolumn{1}{l}{$478692844$} & \multicolumn{1}{c}{{\color{black} $\approx 1.86$}} & \multicolumn{1}{c}{{\color{black}}} & \multicolumn{1}{c}{{\color{black}}}\\
			\multicolumn{1}{c}{} & \multicolumn{1}{c}{{\color{black} 33}} & \multicolumn{1}{l}{$423236754$} & \multicolumn{1}{c}{{\color{black} $\approx 1.82$}} & \multicolumn{1}{c}{{\color{black}}} & \multicolumn{1}{c}{{\color{black}}}\\
			\\[-2.5ex]\hline\\[-2ex]
			\multicolumn{1}{c}{{\color{black} 12}} & \multicolumn{1}{c}{{\color{black} 34}} & \multicolumn{1}{l}{$773361622$} & \multicolumn{1}{c}{{\color{black} $\approx 1.82$}} & \multicolumn{1}{c}{{\color{black}}} & \multicolumn{1}{c}{{\color{black}}}\\
			\multicolumn{1}{c}{} & \multicolumn{1}{c}{{\color{black} 34}} & \multicolumn{1}{c}{{\color{black}}} & \multicolumn{1}{c}{{\color{black}}} & \multicolumn{1}{c}{{\color{black}}} & \multicolumn{1}{c}{{\color{black}}}\\
			\\[-2.5ex]\hline
		\end{tabular}
		\caption{Number of sub-problems generated by the configurations (entry-wise) presented in Table~\ref{c-tri:sections:experiments:table:nestedTri:1}.}
		\label{c-tri:sections:experiments:table:nestedTri:2}
	\end{sidewaystable}
}

{\arrayrulecolor{black}
	\begin{sidewaystable}
		\centering\small
			\begin{tabular}{ccccccc}
				& & & & & \multicolumn{1}{c}{{\color{black} Time in hhh:mm:ss.ms}} & \multicolumn{1}{c}{{\color{black}RAM in Mb}}\\ 
				\\[-2.5ex]\cline{5-7}\\[-2ex]
				$n$ & $m$ & $k$ & {\color{black}\#Triangulations} & {\color{black} Base} & {\color{black} sn-paths} & {\color{black} sn-paths}\\
				\\[-2.5ex]\hline\\[-2ex]
					\gridOne{6}{6}{260420548144996}{\approx 2.51}{10.36}{}{16}{}{3}
					\\[-2.5ex]\hline\\[-2ex]
					\gridOne{6}{7}{341816489625522032}{\approx 2.61}{45.79}{}{41}{}{3}
					\\[-2.5ex]\hline\\[-2ex]
					\gridOne{6}{8}{464476385680935656240}{\approx 2.69}{2:31.70}{}{107}{}{3}
					\\[-2.5ex]\hline\\[-2ex]
					\gridOne{6}{9}{645855159466371391947660}{\approx 2.76}{7:28.22}{}{213}{}{3}
					\\[-2.5ex]\hline\\[-2ex]
					\gridOne{6}{10}{913036902513499041820702784}{\approx 2.81}{17:27.63}{}{460}{}{3}
					\\[-2.5ex]\hline\\[-2ex]
					\gridOne{6}{11}{1306520849733616781789190513820}{\approx 2.85}{39:35.88}{}{840}{}{3}
					\\[-2.5ex]\hline\\[-2ex]
					\gridOne{6}{12}{1887591165891651253904039432371172}{\approx 2.89}{1:22:15}{}{1555}{}{3}
					\\[-2.5ex]\hline\\[-2ex]
					\gridOne{6}{13}{2747848427721241461905176361078147168}{\approx 2.93}{2:50:42}{}{2479}{}{3}
					\\[-2.5ex]\hline\\[-2ex]
					\gridOne{6}{14}{4024758386310801427793602374466243714608}{\approx 2.96}{5:14:23}{}{4439}{}{3}
					\\[-2.5ex]\hline\\[-2ex]
					\gridOne{6}{15}{5924744736041718687622958191829471010847132}{\approx 2.98}{8:43:14}{}{6315}{}{3}
					\\[-2.5ex]\hline\\[-2ex]
					\gridOne{6}{16}{8757956199571261116690226598764501142088496860}{\approx 3.01}{14:55:29}{}{10344}{}{3}
					\\[-2.5ex]\hline\\[-2ex]
					\gridOne{6}{17}{12991215957916577635251095613859465176216530106080}{\approx 3.03}{26:08:38}{}{15023}{}{3}
					\\[-2.5ex]\hline\\[-2ex]
					\gridOne{7}{7}{1999206934751133055518}{\approx 2.72}{6:09.75}{}{187}{}{4}
					\\[-2.5ex]\hline\\[-2ex]
					\gridOne{7}{8}{12169409954141988707186052}{\approx 2.80}{36:59.53}{}{869}{}{4}
					\\[-2.5ex]\hline\\[-2ex]
					\gridOne{7}{9}{76083336332947513655554918994}{\approx 2.87}{2:28:42}{}{2344}{}{4}
					\\[-2.5ex]\hline\\[-2ex]
					\gridOne{7}{10}{484772512167266688498399632918196}{\approx 2.93}{8:12:59}{}{6465}{}{4}
					\\[-2.5ex]\hline\\[-2ex]
					\gridOne{7}{11}{3131521959869770128138491287826065904}{\approx 2.97}{27:42:55}{}{14870}{}{4}
					\\[-2.5ex]\hline\\[-2ex]
					\gridOne{7}{12}{20443767611927599823217291769468449488548}{\approx 3.01}{67:06:41}{}{34752}{}{4}
					\\[-2.5ex]\hline\\[-2ex]
					\gridOne{8}{8}{332633840844113103751597995920}{\approx 2.89}{6:00:49}{}{6171}{}{4}
					\\[-2.5ex]\hline\\[-2ex]
					\gridOne{8}{9}{9369363517501208819530429967280708}{\approx 2.96}{32:49:11}{}{19071}{}{4}
					\\[-2.5ex]\hline\\[-2ex]
					\gridOne{8}{10}{269621109753732518252493257828413137272}{\approx 3.02}{139:58:01}{}{75220}{}{4}
				\\[-2.5ex]\hline
			\end{tabular}
			\caption{Grid of $n$ by $m$ with $k$ onion layers.}
			\label{c-tri:sections:experiments:table:grid:1}
	\end{sidewaystable}
	
%	\begin{sidewaystable}[p]
	\begin{table}
		\centering\small
			\begin{tabular}{cccccc}
				& & & \multicolumn{2}{c}{{\color{black} \# Sub-problems}}\\
				\\[-2.5ex]\cline{4-6}\\[-2ex]
				$n$ & $m$ & $k$ & {\color{black} sn-paths} & {\color{black} Base} & {\color{black} Exp}\\
				\\[-2.5ex]\hline\\[-2ex]
					\gridTwo{6}{6}{69908}{\approx 1.36}{}{}{3}{\approx 1.03}
					\\[-2.5ex]\hline\\[-2ex]
					\gridTwo{6}{7}{207193}{\approx 1.33}{}{}{3}{\approx 1.09}
					\\[-2.5ex]\hline\\[-2ex]
					\gridTwo{6}{8}{465416}{\approx 1.31}{}{}{3}{\approx 1.12}
					\\[-2.5ex]\hline\\[-2ex]
					\gridTwo{6}{9}{1002029}{\approx 1.29}{}{}{3}{\approx 1.15}
					\\[-2.5ex]\hline\\[-2ex]
					\gridTwo{6}{10}{1883205}{\approx 1.27}{}{}{3}{\approx 1.17}
					\\[-2.5ex]\hline\\[-2ex]
					\gridTwo{6}{11}{3409331}{\approx 1.25}{}{}{3}{\approx 1.19}
					\\[-2.5ex]\hline\\[-2ex]
					\gridTwo{6}{12}{5705962}{\approx 1.24}{}{}{3}{\approx 1.21}
					\\[-2.5ex]\hline\\[-2ex]
					\gridTwo{6}{13}{9417222}{\approx 1.22}{}{}{3}{\approx 1.22}
					\\[-2.5ex]\hline\\[-2ex]
					\gridTwo{6}{14}{14471156}{\approx 1.21}{}{}{3}{\approx 1.24}
					\\[-2.5ex]\hline\\[-2ex]
					\gridTwo{6}{15}{22201708}{\approx 1.20}{}{}{3}{\approx 1.25}
					\\[-2.5ex]\hline\\[-2ex]
					\gridTwo{6}{16}{32491047}{\approx 1.19}{}{}{3}{\approx 1.26}
					\\[-2.5ex]\hline\\[-2ex]
					\gridTwo{6}{17}{46979052}{\approx 1.18}{}{}{3}{\approx 1.27}
					\\[-2.5ex]\hline\\[-2ex]
					\gridTwo{7}{7}{972496}{\approx 1.32}{}{}{4}{\approx 0.88}
					\\[-2.5ex]\hline\\[-2ex]
					\gridTwo{7}{8}{3527752}{\approx 1.30}{}{}{4}{\approx 0.93}
					\\[-2.5ex]\hline\\[-2ex]
					\gridTwo{7}{9}{10558836}{\approx 1.29}{}{}{4}{\approx 0.97}
					\\[-2.5ex]\hline\\[-2ex]
					\gridTwo{7}{10}{25013282}{\approx 1.27}{}{}{4}{\approx 1}
					\\[-2.5ex]\hline\\[-2ex]
					\gridTwo{7}{11}{55453561}{\approx 1.26}{}{}{4}{\approx 1.02}
					\\[-2.5ex]\hline\\[-2ex]
					\gridTwo{7}{12}{109901193}{\approx 1.24}{}{}{4}{\approx 1.04}
					\\[-2.5ex]\hline\\[-2ex]
					\gridTwo{8}{8}{14569428}{\approx 1.29}{}{}{4}{\approx 0.99}
					\\[-2.5ex]\hline\\[-2ex]
					\gridTwo{8}{9}{50333235}{\approx 1.27}{}{}{4}{\approx 1.03}
					\\[-2.5ex]\hline\\[-2ex]
					\gridTwo{8}{10}{122283519}{\approx 1.26}{}{}{4}{\approx 1.06}
				\\[-2.5ex]\hline
			\end{tabular}
			\caption{Number of sub-problems generated by the configurations (entry-wise) presented in Table~\ref{c-tri:sections:experiments:table:grid:1}.}
			\label{c-tri:sections:experiments:table:grid:2}
	\end{table}
%	\end{sidewaystable}
}

The experiments turned out to be what we had expected, namely, generally worse behavior as the number of onion layers increases, since the number of triangulations should certainly increase with the number of onion layers. The experiments show however that \emph{both} algorithms are counting triangulations by generating far fewer sub-problems: Having a glimpse at the number of sub-problems in Tables~\ref{c-tri:sections:experiments:table:3cl:1} to~\ref{c-tri:sections:experiments:table:grid:2}, each looks as something of the sort $\sqrt{|\F_{T}(\setp)|}$, which was already reported in~\cite{ray-seidel} for the ray-seidel algorithm.

The ray-seidel algorithm showed a consistent behavior across all experiments. This algorithm lived up to its expectations, it turned out to be simply the fastest algorithm, but this came with the price of being \emph{very} resource-consuming. We can see in the tables that the resources the algorithm uses increase very fast, at that rate we could say that increasing RAM to a couple of Terabytes will not really allow us to run the algorithm on significantly larger set of points. However, there are other techniques we could use to alleviate this situation, we could for example decide to store only a subset of the produced sub-problems and re-compute a sub-problem whenever needed. Since apparently computing sub-problems is very fast, we could expect that this method does not severely blow up the running time.

Now turning to sn-paths, the algorithm really performed best for three convex layers, see Tables~\ref{c-tri:sections:experiments:table:3cl:1} and~\ref{c-tri:sections:experiments:table:3cl:2}. For those configurations the algorithm allowed us to go up to 80 points in a ``reasonable'' amount of time without exhausting the RAM, which is almost twice as much as the ray-seidel algorithm allowed. In this regard we believe that increasing computational power and resources, say 512 GB of RAM, could allow us to go somewhere near 160 points. The ray-seidel algorithm would get nowhere close to this number of points \emph{per se}. This ``nice'' behavior can also be seen in Tables~\ref{c-tri:sections:experiments:table:grid:1} and~\ref{c-tri:sections:experiments:table:grid:2} with grids having three and four onion layers, however, grids are believed to have far less triangulations than sets of points in general position. This is also supported by the experiments. For fewer than three onion layers the algorithm gets better, so they are really not an issue. Also, since the running time of sn-paths can be expressed as $n^{d\cdot k}$, for some positive $d$, the idea of column Exp in the tables was to see whether that value comes out roughly as a constant for the same values of $k$, however, the data set seems to be small to show this. We believe that the right value should be $3\leq d\leq 4$. Finally, by increasing the number of onion layers, Tables~\ref{c-tri:sections:experiments:table:square:1} to~\ref{c-tri:sections:experiments:table:nestedTri:2}, we can see how the behavior of sn-paths quickly deteriorates, in all aspects.

The conclusions of the experiments really suggest themselves. In the ``low'' end, up to 25 points, any algorithm will do but ray-seidel is the fastest. In the ``high'' end it really seems that sn-paths is a better choice due to its smaller memory footprint, so up to 6--7 onion layers we would stick with it, but beyond that number of onion layers we would consider ray-seidel a better option.

\section{Conclusions}\label{c-tri:sections:conclusionsTripleP}

In this paper we have presented algorithms to count triangulations, crossing-free matchings and crossing-free spanning cycles of a given set of points $\setp$. All algorithms use the onion layers of $\setp$ and the divide-and-conquer paradigm.

The algorithm to count triangulations presented in this paper has the best provable worst-case running time as of this writing, $O^{*}(3.1414^{n})$. We consider important to note that configurations of points having from $\Omega\left(3.464^{n}\right)$, see~\cite{DBLP:journals/jct/SantosS03a}, to $\Omega\left(8.65^{n}\right)$, see~\cite{DBLP:conf/stacs/DumitrescuSST11}, triangulations are known. Thus the algorithm presented in this paper counts triangulations faster than enumeration algorithms in at least those cases. Also, this algorithm has polynomial-time instances whenever the number of onion layers of the given set of points is constant. As we saw in the experiments we showed, for up to 3 onion layers this algorithm outperforms the algorithm of~\cite{ray-seidel}, which is reported to be extremely fast in practice. At this point the most interesting open questions w.r.t.~counting triangulations are the following: (\oldstylenums{1}) Is it possible to \emph{always} count triangulations in polynomial time? or is this problem \#P-complete? (\oldstylenums{2}) Is it true that \emph{every} set of $n$ points, $n$ being large enough, spans at least $\Omega\left(3.464^{n}\right)$ triangulations? If this is true, then the algorithm we presented in this paper \emph{always} counts triangulations in $o\left(|\F_{T}(\setp)|\right)$. 

Speaking about counting other kinds of crossing-free structures, we showed again two algorithms. These algorithms could be seen more as a framework for counting essentially ``every'' kind of crossing-free structures, since it depends on a labeling scheme, which is the hardest part of the algorithm to come up with. This ``framework'' implies algorithms with running times of the sort $n^{O(k)}$, where $k$ is the number of onion layers of the given set of points, which again, for fixed $k$ implies polynomial time. Algorithms like these were not known before for this kind of problems. This gives a partial answer to Problem 16 of The Open Problems Project, which asks whether $|\F_{C}(\setp)|$ can \emph{always} be computed in polynomial time, see~\cite{topp}. These counting algorithms also allow us to generate crossing-free matchings and spanning cycles uniformly at random. The latter being a problem that has attracted the attention of researchers for almost 20 years, in the form of generating random simple polygons on $\setp$, which is nothing but a crossing-free spanning cycle of $\setp$. Since our algorithms are based on the divide-and-conquer paradigm, we can adapt the method explained in~\cite{DBLP:conf/compgeom/Aichholzer99} to produce such random structures, once the counting has been done. Other methods to generate random simple polygons without having to count are known, for example, in~\cite{DBLP:conf/cccg/AuerH96} many heuristics for polygons are presented. There the authors reported that (uniform) random generation can be done in polynomial time when the random polygon is star-shaped, but in the general case the algorithms therein presented are either unpractical or unable to generate uniformly at random. 

It is worth noting that although we could have tried to come up with an annotation scheme for pseudo-triangulations, and thus we could have obtained an algorithm to count pseudo-triangulations using the onion layers, the resulting algorithm would have had a running time of the sort $n^{O(k)}$. This running time is in general not better than that of the algorithm presented in~\cite{sweep-line} for pseudo-triangulations, at least from the theoretical point of view, since for the latter we have $O^{*}\left(c^{n}\right)$, for some constant $c\in\R$, while for the former $k$ can get linearly large. The most important open problem here is whether the number of matchings and spanning cycles can \emph{always} be computed in polynomial time.

We also showed a hardness result of a very particular instance of the problem of counting triangulations exactly, namely the {\sc Restricted-Triangulation-Counting-Problem} (RTCP). We showed that this problem is W[2]-hard if the parameter is the number of onion layers of the set of points it is defined on. The algorithm for counting triangulations of Theorem~\ref{c-tri:theorems:triple-paths} needs little to no modifications to be run on instances of RTCP, and the separation FPT $\neq$ W[2] is widely believed, so we can still hope that, by exploiting structural properties of triangulations, and also of other crossing-free structures, we can obtain FPT algorithms for the counting problems studied in~\S~\ref{c-tri:sections:triple-paths} and~\S~\ref{c-tri:sections:triangular-paths}. Thus, another interesting question at this moment is: Do those counting problems belong to FPT or not?

Finally, we showed experimental results comparing the algorithms for counting triangulations of Theorem~\ref{c-tri:theorems:triple-paths} and the algorithm of~\cite{ray-seidel}. Those experiments give a rough idea of what to expect when running each one of those algorithms on real configurations of points. It would be very interesting to see a hybrid algorithm that uses the sn-path and the ray-seidel algorithms, if possible. That algorithm could combine the small memory footprint of the sn-path algorithm with the fast execution of the ray-seidel algorithm. This could allow us to solve larger sets of points.

\section{Acknowledgement}

We thank Raimund Seidel for valuable feedback and interesting discussions.

%%-This prints all the tables before the bibliography.
\clearpage

\bibliographystyle{ieeetr}
\bibliography{bibliography-crossing-free-structures}
\end{document}